\documentclass[11pt,a4paper]{article}
\usepackage[mathscr]{eucal} 
\usepackage{amsmath,amsfonts,amsthm}
\usepackage{mathtools}
\newtheorem{theorem}{Theorem}[section]
\newtheorem{lemma}[theorem]{Lemma}
\newtheorem{proposition}[theorem]{Proposition}
\newtheorem{corollary}[theorem]{Corollary}
\theoremstyle{definition}
\newtheorem{definition}[theorem]{Definition}

\newtheorem{remark}[theorem]{Remark}

\numberwithin{equation}{section}

\newcommand\cN{{\mathcal N}}

\newcommand\scM{{\mathscr M}}
\newcommand\scN{{\mathscr N}}
\newcommand\scO{{\mathscr O}}

\newcommand\scS{{\mathscr S}}

\newcommand\scW{{\mathscr W}}

\newcommand\mvector{\boldsymbol}

\newcommand\va{\mvector{a}}
\newcommand\vb{\mvector{b}}

\newcommand\vd{\mvector{d}}

\newcommand\vm{\mvector{m}}

\newcommand\vp{\mvector{p}}

\newcommand\vr{\mvector{r}}
\newcommand\vs{\mvector{s}}

\newcommand\vu{\mvector{u}}
\newcommand\vvv{\mvector{v}}

\newcommand\vx{\mvector{x}}

\newcommand\vz{\mvector{z}}
\newcommand\vA{\mvector{A}}
\newcommand\vB{\mvector{B}}
\newcommand\vC{\mvector{C}}
\newcommand\vD{\mvector{D}}

\newcommand\vG{\mvector{G}}
\newcommand\vH{\mvector{H}}
\newcommand\vI{\mvector{I}}
\newcommand\vJ{\mvector{J}}

\newcommand\vM{\mvector{M}}

\newcommand\vP{\mvector{P}}

\newcommand\vR{\mvector{R}}

\newcommand\vT{\mvector{T}}

\newcommand\vV{\mvector{V}}

\newcommand\vX{\mvector{X}}
\newcommand\vY{\mvector{Y}}

\newcommand\vzero{\mvector{0}}
\newcommand\vveps{\mvector{\varepsilon}}
\newcommand\field{\mathbb}

\newcommand\R{\field{R}}

\newcommand\CC{\field{C}}

\newcommand\N{\field{N}}

\newcommand\diag{\operatorname{diag}}

\newcommand\rmd{\mathrm{d}}

\newcommand\rmi{\mathrm{i}\mspace{1mu}}
\newcommand\rme{\mathrm{e}}
\newcommand\Dt{\frac{\mathrm{d}\phantom{t} }{\mathrm{d}\mspace{1mu}
t}}

\newcommand\pder[2]{\dfrac{\partial #1 }{\partial #2}}

\newcommand\abs[1]{\lvert #1 \rvert}

\newcommand\mtext[1]{\quad\text{#1}\quad}

\newcommand\defset[2]{\left\{{#1}\;\vert \;\; {#2} \,\right\}}

\DeclareRobustCommand*{\sump}{%
    \mathop{{\sum}^{\mathrlap{\prime}}}%
}
\usepackage[square,numbers]{natbib}
\title{Non-integrability of the $n$-body problem}

\author{
Andrzej J. Maciejewski\\
Janusz Gil Institute of Astronomy, University of Zielona G\'ora, \\
ul. Licealna~9, 65-417,  Zielona G\'ora, Poland \\
e-mail: a.maciejewski@ia.uz.zgora.pl \\ and \\
Maria Przybylska \\
Institute of Physics, University of Zielona G\'ora, \\
 Licealna 9, PL-65--417,  Zielona G\'ora, Poland
\\ e-mail: M.Przybylska@if.uz.zgora.pl \\ and \\
Thierry Combot\\
Universit\'e de Burgogne\\
9 avenue A. Savary BP 47870,21078 Dijon Cedex, France\\
e-mail: thierry.combot@u-bourgogne.fr
}

\begin{document}
\mathtoolsset{
mathic = true
}

\maketitle
\begin{abstract}
We prove that the classical planar $n$-body problem  when restricted to a common
level of the energy and the angular momentum is not integrable except in the
case when both values of these integrals are zero. In the proof of our theorem,
we use methods of differential Galois theory.
\end{abstract}

\section{Introduction}
\label{sec:centr}
The classical $n$-body problem is a mechanical system of $n$ points with
positive masses $m_1, \ldots, m_n$ interacting gravitationally according to Newton's law. We consider the planar version of this problem. In the chosen
inertial frame, the positions and velocities of the points are given by the vectors $\vr_i\in\R^2$ and $\vvv_i=\dot \vr_i\in\R^2$, $1\leq i \leq n$. Then
the potential energy of the system is
\begin{equation}
\label{eq:V}
V(\vr)= -\sum_{1\leq i<j\leq n}\frac{m_i m_j}{\abs{\vr_i - \vr_j}},
\end{equation}
where $\vr=(\vr_1, \ldots, \vr_n)\in \left(\R^2\right)^n$.  We chose units in
such  that the gravitational constant is $G=1$.

The Hamiltonian function of this system has the form
\begin{equation}
\label{eq:ham}
H = \frac{1}{2} \vp^T \vM^{-1}\vp +V(\vr),
\end{equation}
where $\vp=(\vp_1, \ldots, \vp_n)\in \left(\R^2\right)^n$; 
$\vp_i=m_i\dot\vr_i\in\R^2$ for $1\leq i\leq n$, are the momenta, and $\vM$ is
$(2n)\times(2n)$ matrix of the form
\begin{equation}
\vM=
\begin{bmatrix}
m_1 \vI_2 & \vzero & \ldots & \vzero \\
\vzero & m_2 \vI_2 & \ldots & \vzero \\
\hdotsfor{4}\\
\vzero &\vzero & \ldots  &m_n \vI_2
\end{bmatrix}.
\end{equation}
Here, $\vI_2$, and $\vzero$ denote $2\times2$ the identity and zero matrices,
respectively.

The canonical equations of motion
\begin{equation}
\label{eq:eqsH}
\dot\vr = \vM^{-1 }\vp,  \qquad \dot\vp = -\nabla V(\vr),
\end{equation}
admit four first integrals: the Hamiltonian~\eqref{eq:ham},  the linear and the
angular momenta
\begin{equation}
\vp_{\mathrm{s}}=\sum_{i=1}^n \vp_i, \qquad C= \sum_{i=1}^n \vr_i^T\vJ_2^T\vp_i,
\qquad
\vJ_2 =
\begin{bmatrix}
0 & -1 \\
1& 0
\end{bmatrix}.
\end{equation}
These first integrals do not  commute
\begin{equation}
\{C,p_{\mathrm{s},1}\}=p_{\mathrm{s},2}\mtext{and} 
\{C, p_{\mathrm{s},2}\} = -p_{\mathrm{s},1}.
\end{equation}
The system has $2n$ degrees of freedom and for its integrability in the
Liouville sense $2n$ functionally independent and pairwise commuting first
integrals are needed.

We consider the problem in the  centre of mass reference frame, so we assume
that 
\begin{equation}
\label{eq:cm}
\vr_\mathrm{s} = \frac{1}{m}\sum_{i=1}^n m_i\vr_i=\vzero, \qquad \vp_{\mathrm{s}}=\vzero, 
\qquad m = \sum_{i=1}^{n} m_i.
\end{equation}
With these restrictions, the system remains Hamiltonian and has $2(n-1)$
degrees of freedom. 

The study of the integrability of the $n$-body problem has a long history. 
H.~Bruns~\cite{Bruns:1887::}  showed the non-existence of additional algebraic
first integrals, except energy, linear, and angular momentum.  His proof was
improved and generalized in \cite{Julliard-Tosel:00::}.

The first proof of the non-integrability of the fully reduced classical planar
three-body problem was done by A.~Tsygvintsev in
\cite{Tsygvintsev:00::a,Tsygvintsev:01::a}, and later by an application of the
differential Galois approach in \cite{Boucher:03::}, see also
\cite{Morales:09::a,Simon:07::}.  In \cite{mp:10::b} we have found a
surprisingly simple proof of the non-integrability of the classical three-body
problem, and later in \cite{mp:11::a} we generalized this result to the case of
three bodies that attract mutually with the force proportional to a certain
negative integer power of the distance between the bodies.

The first proof of the non-integrability of the planar $n$-body problem for
arbitrary $n>2$ and arbitrary positive masses was done by T.~Combot in~\cite{Combot497649}.

All the above-mentioned proofs are based either on applications of the Ziglin
theory \cite{Ziglin:82::b,Ziglin:83::b}, or Morales--Ramis theory
\cite{Morales:99::c}. In both theories variational equations along a particular
solution are investigated, and, in particular,  necessary conditions for the
integrability are formulated in terms of the monodromy or the differential
Galois groups of these equations. Both theories assume that the system under consideration is complex analytic or meromorphic. If these theories are
applied to the study of the integrability of a real system, in fact, the complexified system is considered.   

As we will use differential Galois methods, we will consider the complexification of the $n$-body problem. Therefore, we assume that $\vr, \vp\in
\left( \CC^2 \right)^n$. A vector $\vx=(x_1, \ldots, x_m)\in \CC^m$ is considered
as one column matrix. We will use the following notation 
\begin{equation}
\label{eq:scalar}
\va\cdot\vb=\va^T\vb=\sum_{i=1}^{m}a_ib_i, \mtext{for}\va, \vb\in \CC^m,
\end{equation}
and $\abs{\vx}=\sqrt{\vx\cdot\vx}$. Here, we underline that~\eqref{eq:scalar} does not define a Hermitian scalar product, so $\abs{\vx}$ is, in general, a complex number.

Moreover, as we already mentioned, the Ziglin and Morales--Ramis theories
require that the considered system, first integral and vector fields, be
meromorphic.  However, the gravitational potential~\eqref{eq:V} is not
meromorphic, thus neither the Hamiltonian \eqref{eq:ham}, nor the Hamiltonian
vector field $\vX_H$ generated by $H$ is meromorphic.   The technique to deal
with this problem was given by T.~Combot in \cite{Combot047475}, see also
\cite{mp:16::a}. We use it also in this paper.

To this end, we take the Cartesian product $\scN = \left( \CC^2 \right)^n \times
\left( \CC^2 \right)^n \times \left( \CC^\ast \right)^{n(n-1)/2}$  with
coordinates $(\vr, \vp, \vd)$ where $\vd$ is a vector of dimension $n(n-1)/2$
with coordinates $d_{i,j}$, for $1\leq i<j\leq n$. The problem of $n$-body can
be considered in the complex symplectic manifold $\widetilde{\scM}\subset \cN$
as defined by
\begin{equation}
\widetilde{\scM}=\defset{(\vr,\vp,\vd)\in \scN}{d_{i,j}^2=|\vr_i-\vr_j|^2,
	\quad 1\leq i<j\leq n }. 
\end{equation}
Now, the Hamiltonian $H$ considered as a function on $\widetilde{\scM}$ is a
rational function of $(\vr,\vp,\vd)$. We will investigate the system in the
centre of mass reference frame, so the phase space of the problem is restricted
to the following symplectic manifold 
\begin{equation}
\scM = \widetilde{\scM} \cap 
\defset{(\vr,\vp,\vd)\in\scN}{\vr_{\mathrm{s}}=\vzero, \quad \vp_{\mathrm{s}}=\vzero}.
\end{equation} 
The system restricted to $\scM$ remains Hamiltonian and admits the first two
integrals $H$ and $C$, which are rational on $\scM$, and commute. Therefore, we can study the integrability of the planar $n$-body problem when restricted to a common level of these first integrals, specifically to the manifold.
\begin{equation}
\label{eq:Mhc}
\scM_{h,c}=\defset{(\vr,\vp,\vd)\in \scM}{h = H(\vr,\vp,\vd), \quad c=C(\vr,\vp)}.
\end{equation}
The precise meaning of integrability on a common level of the first integrals is
given in the next section. 

In this article, our aim is to study the rational  integrability of the $n$-body problem
in $\scM_{h,c}$. Our main result is described in the following theorem. 
\begin{theorem}
	\label{thm:we}
	If $(h,c)\neq (0,0)$, then the planar $n$-body problem, with $n>2$, and positive
	masses $m_1, \ldots, m_n$,  restricted to the level $\scM_{h,c}$, is not
	integrable.
\end{theorem}
The meaning of the integrability on a common level of a known first integral is defined in the next section.

In the proof of this theorem, we use methods of differential Galois theory.
If a complex Hamiltonian system is integrable in the Liouville sense, then
according to the Morales--Ramis theorem \cite{Morales:99::c}, the identity
component of the differential Galois group of variational equations along a
particular solution is Abelian. However, we investigate integrability on a
common level of known first integrals, so the necessary conditions for
integrability are weaker. We show that for this integrability, the necessary condition is
that the identity component of the differential Galois group of the variational
equations along a particular solution is solvable; see the next section.

For the proof of our Theorem~\ref{thm:we}, we need a particular solution. Therefore,
we have to find a particular solution to the $n$-body problem, for arbitrary $n>2$,
for arbitrary masses, and arbitrary values of energy and angular momentum.
For the general $n$-body problem, we have, in fact, only one possibility. Namely,
homothetic or homographic solutions associated with the Euler--Moulton collinear
central configurations. However, we do not know these configurations explicitly;
we only know that they exist. Nevertheless, they have nice properties, which are
crucial for our proof.  

First, the variational equations for the homographic solution associated
with these configurations are split into a direct product of four-dimensional
subsystems. This makes investigations of their differential Galois group
tractable. We show that each of the four-dimensional subsystems splits into
two-dimensional subsystems, and thanks to this fact, we could apply well-known
algorithms. The subsystems depend on the values of energy and angular momentum.
They also depend on the chosen central configuration through the
eigenvalues of the Hessian matrix of the potential evaluated at the central
configuration. These eigenvalues are functions of the masses and coordinates of
the configuration, but we do not know their explicit form. In this last step of
our proof, we have to know the signs of these eigenvalues, but we know thanks to the nice theorem of C. Conley, see
Theorem~3.1 in \cite{Pacella:87::}. 

In the proof, three cases can be distinguished.
\begin{enumerate}
	\item If $c=0$ and $h\neq 0$, then we have at our disposal only homothetic
	solutions. The proof for this case coincides with that contained in
	\cite{Combot497649}, and we do not present it here.
	\item If $c\neq 0$ and $2hc^2\neq-\mu^2$, then in our proof we use the
	homographic solution in which all the bodies move in Keplerian orbits with
	non-zero eccentricity $e^2= 1+ 2h c^2/\mu^2$.
	\item If $2hc^2=-\mu^2$, then there exists a homographic solution in which all bodies move in a circular orbit. However, using it we do not obtain any
	obstacles to the integrability because the variational equation reduces to an
	equation with constant coefficients.  In this case, for proof, we take a
	homographic solution in which all the bodies move in the complex part of the
	phase space.
\end{enumerate} 
The statement of our theorem for the case of equal masses    $m_1=\cdots= m_n$  
was proved in \cite{Combot:12::}.

\section{Notion of integrability on a level of known first integrals}
\label{sec:levelint}
For a Hamiltonian system with $n$ degrees of freedom, the integrability means
the integrability in the Arnold-Liouville sense: the system is integrable
iff admits $n$ independent and Poisson commuting first integrals. When there are
not enough first integrals, as in the $n$-body problem, it can be interesting to
study the Hamiltonian system restricted to certain levels of the known first
integrals. The restricted system is then no longer Hamiltonian; however, several
properties persist,  allowing us to define a similar notion of
integrability.

\begin{definition}
	\label{def}
	Consider a Hamiltonian system on a symplectic algebraic manifold $\scM$ of
	dimension $2n$, defined by a rational Hamiltonian $H$, admitting $I_1=H,\dots,I_m$
	functionally independent rational Poisson commuting first integrals. 
	Let 
	\begin{equation*}
	\scS=\defset{\vx\in\scM}{I_1(\vx)=0,\dots, I_m(\vx)=0},
	\end{equation*} 
	be their common level and assume that $\scS$ is an irreducible algebraic variety of
	dimension $2n-m$.  We say that $H$ is integrable restricted to $\scS$ if there
	exist $n-m$ rational functions $I_{m+1},\dots I_n$ on $\scM$ such that
	\begin{itemize}
		\item Integrals  $I_1,\dots I_n$ Poisson commute on $\scS$, that is, $\forall i,j,
		\{I_i,I_j\}=0$ on $\scS$.
		\item Functions $I_{1},\dots I_n$ are independent on $\scS$, i.e., the Jacobian
		matrix
		\begin{equation}\label{eqjac}
		\mathrm{Jac}(I_1, \ldots, I_n)=\begin{bmatrix}
		\partial_{x_1} I_{1} & \dots & \partial_{x_{2n}} I_{1} \\
		\dots & &  \\
		\partial_{x_1} I_n & \dots & \partial_{2n} I_n 
		\end{bmatrix}
		\end{equation}
		has  maximal rank $n$ on an open and dense subset of $\scS$.
	\end{itemize}
\end{definition}
The integrability defined above will also be referred to as restricted integrability. 

Note that the choice of zero levels of integrals $I_1,\dots,I_m$ is not
restrictive. Other levels could be reached by adding arbitrary constants to
$I_i$, and this does not affect the additional conditions for the restricted
integrability. Notice that the second requirement  gives conditions not only
for the restricted integrals $I_{m+1},\dots,I_n$ but also on $I_1,\dots,I_m$.
In fact, their gradients must be linearly independent. This requires that the ideal defined by the numerators of $I_1,\dots,I_m$ is radical. If
not, by the Ziglin lemma \cite{Ziglin:82::b}, we find an algebraic transformation of the first integrals $I_1,\dots,I_m$ 
which makes their gradient independent on $\scS$. The additional first integrals
$I_{m+1},\dots, I_n$ are presented as defined in $\scM$, although we are only
interested in their properties in $\scS$.

\begin{proposition}
	The conditions on $I_{m+1},\dots, I_n$ imposed in Definition~\ref{def} depend
	only on their values on $\scS$.
\end{proposition}

\begin{proof}
	We only have to check that the addition of $I_{m+1},\dots, I_n$ arbitrary functions that disappear in $\scS$ does not affect the conditions. Any rational function vanishing on $\scS$ can be written as
	\begin{equation*}
	\Omega_1 I_1+\dots+ \Omega_m I_m,
	\end{equation*}
	where $\Omega_i$ are rational functions on $\scM$ and smooth on an open set of $\scS$. Now, on $\scS$ we
	have 
	\begin{gather*}
	\Big\{I_i+\sum_{k=1}^{m}\Omega_{i,k} I_k,
	I_j+\sum_{l=1}^{m}\Omega_{j,l} I_l\Big\} = 
	\{I_i,I_j\}+  \sum_{l=1}^m \Omega_{j,l} \{I_i,I_l\} + 
	\sum_{k=1}^m \Omega_{i,k} \{I_k,I_j\} +\\ 
	\sum\limits_{1 \leq k,l \leq m} \Omega_{i,k}\Omega_{j,l} \{I_{k},I_{l}\}=0
	\end{gather*}
	because, by assumptions, all brackets $\{I_i,I_j\}$ vanish on $\scS$ as rational functions, and $\Omega_i$ are smooth on an open set of $\scS$, and thus, this quantity vanishes as a rational function.
	
	Finally, we show that  on $\scS$ Jacobian matrix 
	\begin{equation}
	\mathrm{J}= \mathrm{Jac}\left(I_1, \ldots, I_m, I_{m+1}+ \sum_{k=1}^{m}\Omega_{m+1,k} I_k,
	\ldots, I_n+\sum_{k=1}^{m}\Omega_{n,k} I_k\right),
	\end{equation}
	has the same rank as matrix $ \mathrm{Jac}(I_1, \ldots, I_n)$.  In fact, on
	$\scS$ we have
	\begin{equation}
	\mathrm{J} =
	\begin{bmatrix}
	\partial_{x_1} I_{1} & \dots & \partial_{x_{2n}} I_{1} \\
	\hdotsfor{3} \\
	\partial_{x_1} I_{m} & \dots & \partial_{x_{2n}} I_{m}  \\ 
	\partial_{x_1} I_{m+1} +\sum_{k=1}^m \Omega_{m+1,k}\partial_{x_1} I_{k}& 
	\dots & \partial_{x_{2n}} I_{m+1} + 
	\sum_{k=1}^m \Omega_{m+1,k}\partial_{x_{2n}} I_{k}\\ 
	\hdotsfor{3}\\
	\partial_{x_1} I_{n} +\sum_{k=1}^m \Omega_{n,k}\partial_{x_1} I_{k}& 
	\dots & \partial_{x_{2n}} I_{n} + 
	\sum_{k=1}^m \Omega_{n,k}\partial_{x_{2n}} I_{k}\\ 
	\end{bmatrix}.
	\end{equation}
	The sums in the rows with numbers $j=m+1, \ldots, n$ are just linear combinations of
	the first $m$ rows. This proves our claim.
\end{proof}

\begin{remark}
	\label{rem:prime}
	The condition $\scS$ to be  an irreducible algebraic variety is not always satisfied.
	For example, the Hamiltonian system
	\[
 H=\frac{1}{2}(p_1^2+p_2^2)-\frac{1}{2}(q_1^2+q_2^2)^2,
 \]
	admits the first integral $C=p_1q_2-p_2q_1$. Now, the common level $H=0,C=0$ is not   an irreducible algebraic variety, as the ideal that defines its factors in two components.
	\begin{equation*}
	\begin{split}
	\scS&=\langle p_1q_2-p_2q_1,p_2^2-q_1^2q_2^2-q_2^4,p_1p_2-q_1^3q_2-q_1q_2^3,p_1^2-q_1^4-q_1^2q_2^2\rangle \\ &\cup \langle q_1^2+q_2^2,p_1^2+p_2^2,p_1q_2-p_2q_1,p_1q_1+p_2q_2\rangle.
	\end{split}
	\end{equation*}
	We can, in principle, choose for $\scS$ the algebraic variety defined by one
	of these prime ideals instead of $H=0,C=0$, and then the previous
	proposition follows with little modification. However, note that $\scS$ is a
	common level of the first integrals and not simply common zeros of
	polynomials is an important fact. Indeed, this property allows one to build
	symplectically associated vector fields, which will then ensure
	integrability by quadratures.
\end{remark}
\begin{proposition}
	Consider a Hamiltonian $H$ that is integrable, restricted to a level of first integrals
	$\scS=\{I_1=0,\dots, I_m=0\}$, with additional first integrals
	$I_{m+1},\dots, I_n$.  Then there exist $n$ vector fields defined by $\vX_i(f)= \{I_i,f\}$ that satisfy on $\scS$ the following commutation rules
	\begin{equation*}
	\begin{split}
	&[\vX_i,\vX_j]=0, \quad \text{for}\ i,j\leq m,\quad\text{and}\\
	& [\vX_i,\vX_j]\in \operatorname{Span}_K(\vX_1,\ldots ,\vX_m),\ \quad\text{for}\ i \hbox{ or } j \geq m +1,
	\end{split}
	\end{equation*}
    where $K$ is the field of rational functions on $\scS$. Also, the Hamiltonian vector
	field $\vX_H$ restricted to $\scS$ is integrable by quadratures.
	\label{propo:solvab}
\end{proposition}

Let us first recall the notion of ``distributional integrability''  introduced in article \cite{Carinena:2015::}.

\begin{definition}
Let $V$ be a vector space of vector fields tangential to an irreducible variety $\scN$. We note $\mathcal{D}_V$ the $K$-module of vector fields tangent to $\scN$ generated by $V$, where $K$ is the field of rational functions on $\scN$.\\
The vector space $V$ is called completely regular if for any $p\in\scN$, $V_p=T_p\scN$.\\
For a vector space $W$ of vector fields tangential to $\scN$, we call $W_\ast$ the core of $W$ in $V$, which is the smallest subspace of $V$ such that $W\subset \mathcal{D}_{W_\ast}$.
\end{definition}

Following the construction given in  \cite{Carinena:2015::}, let us  consider a vector field $\Gamma$ tangential to $\scN$, and build the following sequence of vector fields
\[
V_{\Gamma,0}=V,\quad V_{\Gamma,n}= \langle\Gamma\rangle+ [V_{\Gamma,n-1},V_{\Gamma,n-1}]_\ast.
\]

Now Theorem 9 in article \cite{Carinena:2015::} states that if $V_{\Gamma,n}$ becomes Abelian for a certain $n$, then the vector field $\Gamma$ can be integrated using at most $n+1$ successive quadratures.

\begin{proof}
	First integrals  $I_1,\dots, I_m$,  commute in  whole $\scM$, and thus define the commuting vector fields $\vX_i$ at any level of $I_1,\dots I_m$, and thus, in   particular, at $\scS$. As they are first integrals, they also commute with
	$\vX_H$. The first integrals $I_1,\dots,I_n,H$ are rational, and thus can have singularities on $\scS$, and so we will denote by $\widetilde{\scS}$ an open Zariski dense subset of $\scS$ avoiding them. Now, consider $j\geq m+1$. We have
	\begin{equation*}
	\{I_j,I_i\}=\Omega_{j,i,1} I_1+ \ldots +\Omega_{j,i,m} I_m,
	\end{equation*}
	for some smooth functions $\Omega_{j,i,k}$ on $\widetilde{\scS}$, because these Poisson brackets
	should vanish on $\scS$. 
	As 
	\begin{equation}
	[\vX_i,\vX_j]=\vX_{\{I_j,I_i\}},
	\end{equation} 
	we have 
	\begin{equation}
	[\vX_i,\vX_j](f)= \Big\{f,  \Omega_{j,i,1} I_1+ \cdots +
	\Omega_{j,i,m} I_m \Big\}.
	\end{equation}
	Hence, computing this commutator on $\scS$ we get 
	\begin{equation}
	\begin{split}
	[\vX_i,\vX_j](f)= &\Omega_{j,i,1}\{f,I_1\} +\cdots+\Omega_{j,i,m}\{f,I_m\}\\
	=& \Omega_{j,i,1} \vX_1(f)+\cdots+\Omega_{j,i,m}\vX_m(f).
	\end{split}
	\end{equation}
	We thus obtained a Lie algebra of vector field $\mathfrak{g}=\langle
	\vX_1,\dots,\vX_n\rangle$, all tangent to $\scS$, and its derivative is
	\begin{equation*}
	\mathcal{D}(\mathfrak{g}):=[\mathfrak{g},\mathfrak{g}] 
    \subset \hbox{Span}_K (\vX_1,\dots,\vX_m),
	\end{equation*}
	as for $i,j\leq m$  we have $[\vX_i,\vX_j]=\vzero$. Thus by construction, taking the core of this derivative gives
    \[
    [\mathfrak{g},\mathfrak{g}]_\ast \subset \langle\vX_1,\dots,\vX_m\rangle,
    \]
    and so
    \[
    \langle\vX_H\rangle+ [\mathfrak{g},\mathfrak{g}]_\ast \subset  \langle\vX_H,\vX_1,\dots,\vX_m\rangle.
    \]
    Now these vector fields $\vX_H,\vX_1,\dots,\vX_m$ all pairwise commute, and thus form an Abelian algebra. As the first integrals $I_1,\dots,I_n$ are functionally independent, the Lie algebra $\mathfrak{g}=\langle \vX_1,\dots,\vX_n\rangle$ spans, at each regular point, a $n$ dimensional vector space.\\
    Now considering a (regular) leaf $\scN$ of the algebraic foliation given by the common levels of $(I_{m+1},\dots,I_{n})$, we have that $\mathfrak{g}$ is a Lie algebra of vector fields, tangent to $\scN$. As $\dim \scN= 2n-m-(n-m)=n$, the Lie algebra $\mathfrak{g}$ is completely regular on $\scN$. Theorem 9 of \cite{Carinena:2015::} implies then that $\vX_H$, restricted to $\scN$, is solvable by quadratures using at most $2$ successive quadratures. As this holds for almost any leaf $\scN$, $\vX_H$, restricted to $\scS$, is solvable by quadratures.
\end{proof}

This procedure could be generalized by considering towers of restricted first
integrals, still leading to a solvable Lie algebra but with higher solvability
depth. Proposition~\ref{propo:solvab} can now be used to apply Casale's Theorem
3.6 in \cite{Casale:09::}, giving the following corollary.
\begin{corollary}
	\label{hardcore}
Consider a rational Hamiltonian $H$ on a symplectic manifold $\scM$ admitting functionally independent Poisson commuting rational first
	integrals $I_1,\dots,I_m$, and assume that the common level $\scS=\{I_1=0,\dots, I_m=0\}$ defines
	an irreducible algebraic variety of dimension $2n-m$. Let  $\Gamma \subset \scS$ be an algebraic orbit of the vector field $\vX_H$, associated with $H$,   $K$  be the differential field of rational functions on $\Gamma$ equipped with the derivation induced by the vector field $\vX_H$, and $G$ be the differential Galois group over $K$ of order $k$ variational equations along
	$\Gamma$.       If $H$ is integrable
	restricted to $\scS$, then the identity
	component of $G$   is solvable for all $k\in\mathbb{N}$.
\end{corollary}

We will apply this corollary to the $n$-body problem restricted to the common
level of the first integrals $H,C$,  denoted by $\scM_{h,c}$. For this, we need to
ensure that this common level is indeed a irreducible algebraic variety.

\begin{proposition}
	The algebraic variety $\scM_{h,c}$ is   an irreducible variety for $n\geq
	3$, all $h,c\in\mathbb{C}$ and any positive masses. 
\end{proposition}

\begin{proof}
	The variety $\scM_{h,c}$ is defined as $\widetilde{\scM}$ with six
	additional restrictions
	\begin{equation*}
	\vr_s=0,\quad \vp_{\mathrm{s}}(\vp)=0,\quad C(\vr,\vp)-c=0,\quad H(\vr,\vp,\vd)-h=0.
	\end{equation*}
	These conditions, together with the definition of mutual distances $\vd$, define an
	algebraic variety, which is the zero locus of an ideal $\mathcal{I}\subset
	\mathbb{R}[\vr,\vp,\vd]$. We now want to prove that this ideal is prime. We will first prove that this ideal is prime as an ideal of $\overline{\mathbb{C}(\vr)}[\vp]$, which is equivalent to considering
	$\vr,\vd$ as parameters.
	
	Now the restriction $\vr_s=0$ and the definition of mutual distances involve only
	the positions, thus the only relations involving $\vp$ are
	\begin{equation*}
	\vp_{\mathrm{s}}(\vp)=0,\quad C(\vr,\vp)-c=0,\quad H(\vr,\vp,\vd)-h=0.
	\end{equation*}
	The first two are linear
	in $\vp$, and the last is quadratic. If this ideal was not prime, then it would
	factor into two ideals defined by linear equations in $\vp$. Thus, under
	restrictions $\vp_{\mathrm{s}}(\vp)=0$ and $ C(\vr,\vp)-c=0$, the quadratic polynomial
\[
H(\vr,\vp,\vd)-h=L_1(\vr,\vp,\vd)L_2(\vr,\vp,\vd)
\]
would factor in
	$L_1,L_2$ linear in $\vp$. Considering the dominant term in $\vp$, this implies
	that the kinetic part of the Hamiltonian
\[
\frac{1}{2} \vp^T
	\vM^{-1}\vp=\widetilde{L}_1(\vr,\vp,\vd)\widetilde{L}_2(\vr,\vp,\vd)
 \]also factorizes
	into homogeneous linear factors in $\vp$, and thus
 \[\vp^T \vM^{-1}\vp=0,\quad  \forall \vp\in
	\mathcal{V}=\{\vp_{\mathrm{s}}(\vp)=0,C(\vr,\vp)=0,\widetilde{L}_1(\vr,\vp,\vd)=0\}.
 \]
 This kinetic
	part is a positive definite quadratic form, thus after restricting it to the
	real vector space $\vp_{\mathrm{s}}(\vp)=0,C(\vr,\vp)=0$, it is still a positive definite
	quadratic form in $2n-3$ variables. However, the unknown linear factor $\widetilde{L}_1$ could be complex, and thus a priori it is possible that this positive definite
	quadratic form vanishes on $\widetilde{L}_1=0$.

	\begin{lemma}
		\label{lem:quad}
        If a positive definite quadratic form in $m$ variables vanishes on a vector space $\scW$, then $\dim \scW\leq \lfloor \frac{m}{2} \rfloor$.
	\end{lemma}
	
	Here, $\lfloor x \rfloor$ denotes the floor function defined as the greatest integer less than or equal to $x$.
	Using this Lemma for factorization of $\vp^T \vM^{-1}\vp$, we would need 
	\begin{equation*}
	\dim \mathcal{V}=2n-4 \leq \left\lfloor \frac{2n-3}{2} \right\rfloor,
	\end{equation*}
	which is equivalent to $n\leq 2$. Thus, the ideal $\mathcal{I}$ is a prime ideal as
	an ideal of $\overline{\mathbb{C}(\vr)}[\vp]$. To conclude, there could still be
	a factor of $\mathcal{I}$ that contains an element in $\vr,\vd$ only, which would
	thus appear as a constant in $\overline{\mathbb{C}(\vr)}[\vp]$. As previously,
	this would be a factor $L(\vr,\vd)$ of $\vp^T \vM^{-1}\vp$, when restricted to
\[
 \widetilde{\mathcal{V}}= \{\vp_{\mathrm{s}}(\vp)=0,C(\vr,\vp)=0\}.
 \]
	When $L(\vr,\vd)=0$, the quadratic form $\vp^T \vM^{-1}\vp$ would then vanish on $\widetilde{\mathcal{V}}$.
	
	For any given $\vr\in\mathbb{C}^{2n}$, $\vp^T \vM^{-1}\vp$ restricted to
	$\vp_{\mathrm{s}}(\vp)=0$ is a real positive definite quadratic
	form of $2n-2$ variables, and $\widetilde{\mathcal{V}}$ has dimension $2n-3$. Using the Lemma, for
	$\vp^T \vM^{-1}\vp$ vanishing on $\widetilde{\mathcal{V}}$, we need
\[
2n-3\leq  \left\lfloor \frac{2n-2}{2} \right\rfloor, 
\]
	which is equivalent to $n\leq 2$. Thus, $\scM_{h,c}$ is  an irreducible algebraic variety.
\end{proof}

Note that we not only proved that $\scM_{h,c}$ is  an irreducible algebraic variety,
but that the ideal $\mathcal{I}$  is prime. This implies that $\scM_{h,c}$ is
not a multiple level of the first integrals $H$ and $C$, and thus their
gradients are generically linearly independent on $\scM_{h,c}$. Finally, we
recall the following fact. 
\begin{proposition}
	\label{pro:indep}
	Assume that the gradients of $H$ and $C$ are collinear. Then
	$\abs{\vr_i}=\mathrm{const}_i$, for $i=1,\ldots, n$, and $\vp_{\mathrm{s}}=\vzero$.
\end{proposition}
\begin{proof}
	We consider the matrix $[\nabla H, \nabla C]$, and take the following minors \begin{equation}
	\det[ \nabla_{\vp_i} H,  \nabla_{\vp_i} C]=\det[\dot\vr_i, \vJ_2\vr_i]=\vr_i^T\dot\vr_i = \frac{1}{2}\Dt \left(\vr_i^T\vr_i\right)=0.
	\end{equation}
	By our assumption, all these minors vanish, so our first claim follows.
	Now, a direct calculation gives 
	\begin{equation*}
	\frac{1}{2}\Dt \left(\vr_{\mathrm{s}}^T\vr_{\mathrm{s}}\right) = \vr_{\mathrm{s}}^T\dot\vr_{\mathrm{s}} =
	\frac{1}{m^2}\sum_{i=1}^n m_i^2 \vr_i^T\dot\vr_i =0. 
	\end{equation*}
	Therefore, $\abs{\vr_{\mathrm{s}}}=\mathrm{const}$, and this proves the second
	claim.
\end{proof}
\section{Central configurations and homographic solutions}
\label{sec:colcen}
In this section, we recall some facts about central configurations. We
consider only real configurations, assuming in this section that $\vr_i\in\R^2$,
for $i=1,\ldots,n$. 

For further consideration, it is convenient to rewrite the equations of motion
\eqref{eq:eqsH} in the following form
\begin{equation}
\label{eq:newton}
\vM \ddot\vr = - \nabla V(\vr).
\end{equation}
In the $n$-body problem, the points with positions $\vs_i$  form a central
configuration if the total force acting on a point is directed toward the
center of mass. Thus, a vector $ \vs= (\vs_1, \ldots, \vs_n)\in \left(
\R^2\right)^n$ is a central configuration if
\begin{equation}
\label{eq:cc_exp}
\sump_{j=1}^n \frac{m_i m_j}{\abs{\vs_i - \vs_j}^3}\left( \vs_i -
\vs_j\right)
= \mu m_i \vs_i
, \qquad i = 1, \ldots,n,
\end{equation}
for a certain $\mu\in\R$. The primed sum sign denotes the summation over $j\neq
i$. We can rewrite these equations using the potential $ V(\vr)$ given in
\eqref{eq:V} in the form
\begin{equation}
\label{eq:cc}
\nabla V(\vs) = \mu \vM\vs.
\end{equation}
It is easy to show that
\begin{equation}
\label{eq:lam}
\mu := \mu(\vs)=-\frac{V(\vs)}{I(\vs)},  \qquad I(\vs)= \vs^T \vM \vs.
\end{equation}

Notice that if $\vs$ is a central configuration, then, for an arbitrary
$\vA\in\mathrm{SO}(2,\R)$, $\widetilde\vs:=\widehat{\vA}\vs:=(\vA\vs_1, \ldots,
\vA\vs_n)$ is also a central configuration with the same $\mu$. Moreover, if
$\vs$ is a central configuration, then, $\widetilde\vs=\alpha\vs$ is also a
central configuration for an arbitrary $\alpha\neq 0$, with $\mu(
\widetilde\vs)= \alpha^{-3} \mu(\vs)$. Using this property, we can normalize a
central configuration $\vs$ in such a way that $\mu(\vs)=1$. 

Let $\vH(\vr):= \nabla^2 V(\vr)$  be the Hessian matrix of the potential.  It
has the form
\begin{equation}
\label{eq:hess}
\vH(\vr) =
\begin{bmatrix}
\vH_{11}(\vr) & \ldots & \vH_{1n}(\vr)  \\
\hdotsfor{3} \\
\vH_{n1}(\vr) & \ldots & \vH_{nn}(\vr)  \\
\end{bmatrix},
\end{equation}
where $2\times2$ blocks $\vH_{ij}(\vr)$ for $i\neq j$ are given by
\begin{equation}
\label{eq:Hij}
\vH_{ij}(\vr)= -\frac{m_i m_j}{\abs{\vr_i - \vr_j}^3}
\left[ \vI_2 - 3 \frac{(\vr_i-\vr_j)(\vr_i-\vr_j)^T}{\abs{\vr_i - \vr_j}^2}
\right],
\end{equation}
and for $j=i$ are defined as
\begin{equation}
\label{eq:Hii}
\vH_{ii} (\vr)= -\sump_{j=1}^n \vH_{ij}(\vr).
\end{equation}
From the above formulae we can easily deduce that 
\begin{equation}
\label{eq:Hijprop}
\vH_{ij}(\widehat{\vA}\vr)=  \vA \vH_{ij}(\vr)\vA^T,
\end{equation}
and thus
\begin{equation}
\label{eq:H}
\vH(\widehat{\vA}\vr)=\widehat{\vA} \vH(\vr)\widehat{\vA}^T.
\end{equation} 

A central configuration gives a family of particular solutions of the problem
which we will call the Kepler homographic solution, see~\cite{Moeckel:15::}.
Below we will assume that 
\begin{equation}
     \vA(\nu) = 
	\begin{bmatrix}
	\cos \nu & -\sin \nu\\
	\sin \nu &  \cos \nu
	\end{bmatrix}.
\end{equation}
\begin{proposition}
	\label{pro:homographic}
	Let $ \vs= (\vs_1, \ldots, \vs_n)$ be a central configuration. Then
	equations~\eqref{eq:newton} admit solution
	\begin{equation}
	\label{eq:homgr}
	\vr(t) = \rho(t) \widehat\vA(\nu(t))\vs, 
	\end{equation}
	where $\rho(t)$ and $\nu(t)$ are a solution of planar Kepler problem
	\begin{equation}
	\label{eq:kep2}
	\left.
	\begin{aligned}
	& \ddot \rho - \rho \dot\nu^2 = -\frac{\mu}{\rho^2} ,\\
	& \rho\ddot\nu + 2 \dot\rho\dot\nu = 0.
	\end{aligned}
	\quad\right\}
	\end{equation}
\end{proposition}
\begin{proof}
	We substitute  $\vr_i(t) = \rho(t) \vA(\nu(t))\vs_i$ into
	equation~\eqref{eq:newton}. As 
	\begin{equation}
	\label{dotrt}
	\dot\vr_i(t) = \vA(\nu(t)) \left[ \dot\rho \vI_2  + \dot\nu \rho \vJ_2\right]\vs_i,
	\end{equation}
	we get 
	\begin{equation}
	\ddot\vr_i(t) = \vA(\nu(t)) \left[  (\ddot\rho -{\dot\nu}^2\rho) \vI_2 +  ( 2 \dot\rho\dot\nu  +\rho \ddot\nu )\vJ_2\right]\vs_i . 
	\end{equation}
	The right-hand side of equation~\eqref{eq:newton}  after substitution is 
	\begin{multline}
	\label{dVrt}
	-\sump_{j=1}^n \frac{m_i m_j}{\abs{\vr_i - \vr_j}^3}
	\left( \vr_i - \vr_j\right) = -\frac{1}{\rho^2} 
	\vA(\nu(t))\sump_{j=1}^n \frac{m_i m_j}{\abs{\vs_i - \vs_j}^3}
	\left( \vs_i - \vs_j\right) \\ = 
	-\frac{\mu}{\rho^2} m_i\vA(\nu(t))\vs_i,
	\end{multline}
	where in the last equality we used the assumption that $\vs$ is a central
	configuration. Equating both sides  of equation~\eqref{eq:newton} we obtain 
	\begin{equation}
	(\ddot\rho -{\dot\nu}^2\rho) \vs_i + 
	( 2 \dot\rho\dot\nu  +\rho \ddot\nu )\vJ_2\vs_i =  -\frac{\mu}{\rho^2} \vs_i. 
	\end{equation}
	Vectors $\vs_i$ and $ \vJ_2\vs_i  $ are orthogonal, this is why the 
	coefficients of their linear combination vanish simultaneously. This gives equations \eqref{eq:kep2}.
\end{proof}
Later, we consider only these particular solutions for which $\dot\rho\neq0$.  These particular solutions define the phase curves $\Gamma$ 
which are algebraic curves in $\scM$.  By
Proposition~\ref{pro:indep}, along these solutions integrals $H$ and $C$ are
functionally independent.
\section{Variational equations}

We first compute the variational equations in $\widetilde{\scM}$ along the solution $(\vr(t),\vp(t))$ given
by~\eqref{eq:homgr}. They have the form 
\begin{equation}
\label{var1}
\dot\vR=\vM^{-1}\vP, \qquad \dot\vP =-\nabla^2 V(\vr(t))\vR.
\end{equation} 
In further considerations, it will be convenient to work with the system of second-order differential
equations 
\begin{equation}
\label{eq:var_mat}
\vM \ddot\vR = -\nabla^2 V(\vr(t))\vR. 
\end{equation}
It is natural to set $\vR=(\vR_1, \ldots, \vR_n)$,  and 
$\vP=\vM\dot\vR=(\vP_1, \ldots, \vP_n)$, so $\vP_i=
m_i\dot\vR_i$, for $i=1, \ldots, n $ belong to the tangent space of the phase space at a point $(\vr(t),\vp(t))$.

Let us recall that for the planar Kepler problem the energy and angular momentum integrals have the forms 
\begin{equation}
\label{eq:hc}
h = \frac{1}{2} \left({\dot \rho}^2 + \rho^2 {\dot \nu}^2 \right)- 
\frac{\mu}{\rho}, \qquad c = {\rho^2}{\dot \nu}.
\end{equation}
It is convenient to use the true anomaly $\nu$ as the independent variable and this
is why we will assume $c\neq 0$. Moreover, we make the following change of the variables
\begin{equation}
\label{eq:9}
\vR= \rho(\nu) \widehat\vA(\nu) \vM^{-1/2}\vX, \qquad \vX=(\vX_1, \ldots, \vX_n).
\end{equation}
To perform calculations, we need the following relations 
\begin{equation}
\label{eq:8}
\frac{\rmd\phantom{x}}{\rmd t}=\dot\nu \frac{\rmd\phantom{x}}{\rmd \nu},\quad 
\frac{\rmd^{2}\phantom{x}}{\rmd t^2}=
{\dot\nu}^2 \frac{\rmd^{2}\phantom{x}}{\rmd \nu^2} +
\ddot\nu \frac{\rmd\phantom{x}}{\rmd \nu}.
\end{equation}
From equations~\eqref{eq:kep2} we obtain 
\begin{equation}
\dot \nu = \frac{c}{\rho^2}, \qquad 
\ddot \nu = -\frac{2c^2}{\rho^5}\rho', 
\qquad \dot\rho = \dot \nu \rho', \qquad \ddot\rho = {\dot \nu}^2 \rho'' + 
\ddot \nu \rho'
\end{equation}
and 
\begin{equation}
\rho'' = \rho -\frac{\mu \rho^2}{c^2} + 2 \frac{\rho'^2}{\rho}.
\end{equation}
In the above formulae the prime denotes the differentiation with respect to $\nu$. 
Now, direct calculations give
\begin{equation}
\label{eq:10}
\ddot\vR_i=\frac{\mu}{\sqrt{m_i}\rho(\nu)^2} \vA(\nu)\left[
\frac{c^2}{\mu\rho(\nu)}\left(\vX_i''  + 2 \vJ_2\vX_i'\right) -\vX_i  \right].
\end{equation}
Hence,
\begin{multline}
\label{eq:11}
\vM\ddot \vR = \frac{\mu}{\rho(\nu)^2}\widehat\vA(\nu) \vM^{1/2}\left[
\frac{c^2}{\mu\rho(\nu)}\left(\vX''  + 2 \widehat\vJ\vX'\right) 
-\vX \right]\\=
-\frac{1}{\rho(\nu)^2} \widehat\vA(\nu) \vH(\vs) \vM^{-1/2}\vX,  
\end{multline}
where $\widehat\vJ=\diag(\vJ_2, \ldots, \vJ_2) $.  Let us notice also that
\begin{equation}
\label{eq:13}
\vM\widehat\vA(\nu)=\widehat\vA(\nu)\vM \mtext{and} 
\widehat\vJ\widehat\vA(\nu)=\widehat\vA(\nu)\widehat\vJ.
\end{equation}
In effect, the transformed variational equations read
\begin{equation}
\label{eq:14}
f(\nu) \left(\vX''  + 2 \widehat\vJ\vX'\right)=(\vI_{2n} - \widetilde{\vH}(\vs) )\vX,
\end{equation}
where
\begin{equation}
\label{eq:15}
f(\nu):= \frac{c^2}{\mu\rho(\nu)}, \qquad \widetilde{\vH}(\vs) = \frac{1}{\mu}\vM^{-1/2}\vH(\vs)\vM^{-1/2}.
\end{equation}
With the above parametrization we have 
\begin{equation}
\label{eq:P}
\vP = \vM\dot\vR= \frac{c}{\rho^2}\widehat{\vA}(\nu)\vM^{1/2} \left( \rho' \vX +\rho \widehat{\vJ} \vX + \rho\vX'\right).
\end{equation}

F.\,R.~Moulton in \cite{Moulton:1910::} proved that for every ordering of $n$
positive masses, there exists a unique collinear central configuration of $n$ bodies that interact gravitationally. We fix such a central configuration and assume that $\vs_i =(x_i,0)$
for $i=1, \ldots, n$. Now, our objective is to describe the spectrum of the matrix
$\widetilde{\vH}(\vs)$ defined in~\eqref{eq:15}. It has the
same block structure as matrix $\vH(\vs)$, and  
\begin{equation}
\label{eq:HijM}
\widetilde{\vH}_{ij}(\vs)=\frac{1}{\mu \sqrt{m_i m_j}}\vH_{ij}(\vs) = 
C_{ij} \vD, \qquad \vD = 
\begin{bmatrix}
-2  & 0 \\
0  & 1
\end{bmatrix},
\end{equation}
for $i\neq j$
\begin{equation}
C_{ij} = -\frac{\sqrt{m_i m_j}}{\mu\abs{x_i - x_j}^3},
\label{eq:Cij}
\end{equation}
and
\begin{equation}
C_{ii} = -\sump_{j=1}^n\sqrt{\frac{m_j}{m_i} }C_{ij}.
\label{eq:Cii}
\end{equation}
The matrix $\widetilde{\vH}(\vs)$ and the matrix $\vC=[C_{ij}]$ are
symmetric. 
\begin{proposition}
	\label{pro:Heig}
	Vectors $\vM^{1/2}\vs$ and $\vM^{1/2}\widehat{\vJ}\vs$ are eigenvectors of the matrix
	$\widetilde{\vH}(\vs)$ with eigenvalues $-2$ and $1$, respectively.
\end{proposition}
\begin{proof}
	Components of gradient $\nabla V(\vr)$ are homogeneous functions of degree $-2$.
	Thus, by the Euler identity
	\begin{equation}
	\nabla^2V(\vr)\vr = -2 \nabla V(\vr).
	\end{equation}
	Evaluating both sides of this identity at $\vr=\vs$ we obtain 
	\begin{equation}
	\nabla^2V(\vs)\vs = \vH(\vs) \vs =   -2 \nabla V(\vs) = -2 \mu\vM \vs. 
	\end{equation}
	Thus,  
	\begin{equation}
	\widetilde{\vH}(\vs)\vM^{1/2}\vs = \frac{1}{\mu}\vM^{-1/2}\vH(\vs)\vs = 
	\frac{1}{\mu} \vM^{-1/2}\left( -2 \mu\vM \vs\right)= -2 \vM^{1/2}\vs.
	\end{equation}
	We already mentioned that if $\vs$ is a central configuration then $\widehat{\vA}\vs$ is
	also a central configuration for arbitrary $\vA\in\mathrm{SO}(2,\R)$. Therefore, we have the following 
	\begin{equation}
	\nabla V( \widehat{\vA}(\theta)\vs) = \mu \vM \widehat{\vA}(\theta)\vs.
	\end{equation}
 
	Differentiating both sides of this equality at $\theta=0$ we get 
	\begin{equation}
	\nabla^2 V( \vs)\widehat{\vJ}\vs = \mu \vM \widehat{\vJ}\vs.
	\end{equation}
	Hence,
	\begin{equation}
	\widetilde{\vH}(\vs)\vM^{1/2}\widehat{\vJ}\vs =\frac{1}{\mu}\vM^{-1/2}\vH(\vs)\widehat{\vJ}\vs= \vM^{1/2}\widehat{\vJ}\vs.
	\end{equation}
\end{proof}
\begin{proposition}
	Let $\widehat{\va}=(\va, \ldots, \va)\in\R^{2n}$,  where $\va\in\R^2$ is an
	arbitrary vector.  Then $\vM^{1/2}\widehat{\va}$ belongs to the kernel of the matrix $\widetilde{\vH}(\vs)$.
\end{proposition}
\begin{proof}
	It is enough to notice that $\widehat{\va}$  belongs to the kernel of the matrix
	$\vH(\vr)$ for an arbitrary $\vr$, see \eqref{eq:Hij} and \eqref{eq:Hii}.
\end{proof}

Let $\vvv=(v_1, \ldots,v_n)$ be an eigenvector of the matrix $\vC$, defined in~\eqref{eq:Cij} and \eqref{eq:Cii}, with the
corresponding eigenvalue $\lambda$. Then, one can directly check that the vector
$\vvv^{(1)}=(v_1, 0,v_2,0, \ldots, v_n,0)$ is an eigenvector of
$\widetilde{\vH}(\vs)$ with eigenvalue $-2\lambda$. Similarly,
$\vvv^{(2)}=(0,v_1, 0,v_2, \ldots, 0,v_n)$, is an eigenvector of
$\widetilde{\vH}(\vs)$ with eigenvalue $\lambda$. Notice also that
$\widehat{\vJ}\vvv^{(1)}=\vvv^{(2)}$, and $\widehat{\vJ}\vvv^{(2)}=-\vvv^{(1)}$.
Thus, if $\vvv_1, \ldots, \vvv_n$ is a
normalized eigenbasis of matrix $\vC$, then matrix
$\vT=\left[\vvv_1^{(1)},\vvv_1^{(2)},\ldots, \vvv_n^{(1)},\vvv_n^{(2)} \right]$
diagonalizes matrix $\widetilde{\vH}(\vs)$, $\vT\vT^T=\vI_{2n}$. Moreover,
$\vT^T\widehat{\vJ}\vT= \widehat{\vJ}$. 

Now, it is easy to show that  map 
\begin{equation}
\vX =\vT\vY, \qquad \vY = (\vY_1, \ldots, \vY_n),
\end{equation}
transforms equation~\eqref{eq:14} into a direct product of equations 
\begin{equation}
\label{eq:vsplit}
f(\nu)\left(\vY_k''  + 2 \vJ_2\vY_k'\right)=\vG_k\vY_k, 
\qquad  \vG_k = \diag( 3+2 \delta_k, -\delta_k),
\end{equation}
where $\delta_k = \lambda_k  -1$, and $\lambda_1, \ldots, \lambda_n$ are
eigenvalues of the matrix $\vC$. 

We can assume that
\begin{gather}
\vvv_{n}^{(1)}= \vM^{1/2}\widehat{\va},  \qquad \vvv_{n}^{(2)}=\vM^{1/2}\widehat{\vJ}\widehat{\va}, \quad \va=(1,0) , \\
\vvv_{n-1}^{(1)}= \vM^{1/2}\vs, \qquad \vvv_{n-1}^{(2)}=\vM^{1/2}\widehat{\vJ}\vs.
\end{gather}
Now we have to ensure that the sub-system consisting of the first $(n-2)$ equations \eqref{eq:vsplit} describes a variation tangent to $\scM_{h,c}$. Recall that from Proposition \ref{pro:indep}, we know that $\scM_{h,c}$ is a submanifold of $\widetilde{\scM}$ that is smooth along our particular solution.

Let us evaluate  
\begin{equation}
\begin{split}
\nabla H(\vr, \vp) \cdot(\vR,\vP) = &\pder{H}{\vr}\cdot\vR + \pder{H}{\vp}\cdot\vP, \mtext{and}\\
\nabla C(\vr,\vp) \cdot(\vR,\vP) = &\pder{C}{\vr}\cdot\vR + \pder{C}{\vp}\cdot\vP,    
\end{split}  
\end{equation}
for our particular solution and $\vR$ and $\vP$ given by~\eqref{eq:9} and
\eqref{eq:P}, respectively. 
Using equations~\eqref{dVrt} and~\eqref{dotrt}, one gets 
\begin{equation}
\pder{H}{\vr}= \pder{V}{\vr}= \frac{\mu}{\rho^2} \vM\widehat{\vA}(\nu)\vs,
\qquad \pder{H}{\vp}= \dot\vr= \frac{c}{\rho^2}  \widehat{\vA}(\nu) \left[ \rho' \vI_{2n}  + \rho \widehat{\vJ}\right]\vs,
\end{equation}
and 
\begin{equation}
\pder{C}{\vr}= \widehat{\vJ}^T\vp =  \frac{c}{\rho^2} \widehat{\vJ}^T \vM\widehat{\vA}(\nu) \left[ \rho' \vI_{2n}  + \rho \widehat{\vJ}\right]\vs, \qquad \pder{C}{\vp}= \widehat{\vJ}\vr = \rho \widehat{\vJ} \widehat{\vA}(\nu)\vs.
\end{equation}
To proceed we will use the fact that matrices $\widehat{\vJ} $, $\widehat{\vA}$
and $\vM$ pair-wise commute. We get the following
\begin{equation}
\label{eq:gr1}
\pder{H}{\vr}\cdot\vR = \left(\frac{\mu}{\rho^2} \vM\widehat{\vA}(\nu)
\vs \right)\cdot \left(\rho\widehat\vA(\nu) \vM^{-1/2}\vX\right)=
\frac{\mu}{\rho} \left(\vM^{1/2}\vs\right)\cdot\vX = \vvv_{n-1}^{(1)}\cdot\vX
\end{equation}
and 
\begin{equation}
\label{eq:gr2}
\begin{split}
\pder{H}{\vp}\cdot \vP= &\frac{c^2}{\rho^4} 
\left( \widehat{\vA}(\nu) \left[ \rho' \vI_{2n} + \rho \widehat{\vJ}\right]
\vs \right)\cdot \left(\widehat{\vA}(\nu)\vM^{1/2} \left[ \rho' \vX +\rho 
\widehat{\vJ} \vX + \rho\vX'\right]\right) \\
=& \frac{c^2}{\rho^4} \left(  \left[ \rho' \vI_{2n}  + \rho \widehat{\vJ}
\right]\vs \right)\cdot \left(\vM^{1/2} \left[ \rho' \vX +
\rho \widehat{\vJ} \vX + \rho\vX'\right]\right) \\
=& \frac{c^2}{\rho^4} \left(\vM^{1/2}\vs\right)\cdot \left( 
\left[ \rho' \vI_{2n}  + \rho \widehat{\vJ}^T\right]
\left[ \rho' \vX +\rho \widehat{\vJ} \vX + \rho\vX'\right]\right) \\
=&  \frac{c^2}{\rho^4}  \left[ (\rho^2 +\rho'^2)\vvv_{n-1}^{(1)}\cdot\vX + \rho\rho'\vvv_{n-1}^{(1)}\cdot\vX' + \rho^2 \vvv_{n-1}^{(2)}\cdot\vX' 
\right].
\end{split}
\end{equation}
Similarly, we obtain 
\begin{equation}
\label{eq:gr3}
\begin{split}
\pder{C}{\vr}\cdot\vR = &  \frac{c}{\rho}\left[ - 
\rho' \vvv_{n-1}^{(2)}\cdot\vX  + \rho \vvv_{n-1}^{(1)}\cdot\vX\right], \\ 
\pder{C}{\vp}\cdot\vP = &  
\frac{c}{\rho}\left[ \rho' \vvv_{n-1}^{(2)}\cdot \vX + 
\rho \vvv_{n-1}^{(1)}\cdot\vX  +\rho \vvv_{n-1}^{(2)}\cdot\vX' \right].
\end{split}
\end{equation}

\begin{proposition}
	Vector $(\vY,\vY')$  with $\vY_{n-1}=\vY_{n} = \vY_{n-1}'=\vY_{n}'=\vzero$ is
	tangent to $\scM_{h,c}$.
\end{proposition}
\begin{proof}
	We have to show that if  $\vY_{n-1}=\vY_{n} = \vY_{n-1}'=\vY_{n}'=\vzero$, then
	$\nabla H(\vr, \vp) \cdot(\vR,\vP)=0$ and $\nabla C(\vr,\vp) \cdot(\vR,\vP)=0$. 
	We express these conditions in terms of variables $(\vX,\vX')$.
	Notice that 
	\begin{equation}
	\begin{split}
	& \vY_n =(\vvv_n^{(1)}\cdot\vX,\vvv_n^{(2)}\cdot\vX ), \qquad 
	\vY_n' =(\vvv_n^{(1)}\cdot\vX',\vvv_n^{(2)}\cdot\vX' ), \\
	& \vY_{n-1} =(\vvv_{n-1}^{(1)}\cdot\vX,\vvv_{n-1}^{(2)}\cdot\vX ), \qquad 
	\vY_{n-1}' =(\vvv_{n-1}^{(1)}\cdot\vX',\vvv_{n-1}^{(2)}\cdot\vX' ).
	\end{split}
	\end{equation} 
	Inspecting all terms in \eqref{eq:gr1}, \eqref{eq:gr2} and \eqref{eq:gr3} one
	can notice that all of them vanish, so, our claim is proved. 
\end{proof}

The following lemma characterizes the spectral properties of the matrix $\vC$ with entries defined in \eqref{eq:Cij} and \eqref{eq:Cii}. 
\begin{lemma}
	The matrix $\vC$  has real eigenvalues $\lambda_1,
	\ldots, \lambda_n$. Furthermore, $\lambda_n=0$, $\lambda_{n-1}=1$, and
	$\lambda_i>1$ for $i=1, \ldots, n-2$. 
\end{lemma}
\begin{proof}
	Notice that matrix 
	\begin{equation}
	\mu \vm^{-1}\vC\vm, \qquad \vm=\diag(\sqrt{m_1}, \ldots, \sqrt{m_n}),
	\end{equation} 
	coincides with the matrix $\vA$ which appears in Theorem 3.1 in \cite{Pacella:87::}.
	Thus,  our statement follows from this theorem. 
\end{proof}

From the above follows that perturbations such that $\vY_{n-1}=\vY_{n} = \vY_{n-1}'=\vY_{n}'=\vzero$ define a subspace of the variational equation restricted to the tangent of $\scM_{h,c}$, and these perturbations satisfy equations \eqref{eq:vsplit} with $k=1,\dots, n-2$.

\section{Proof of Theorem}
\label{sec:proof} 

Fix the energy value $h$ and the value of angular momentum $c$. As we already
mentioned, if $c=0$ and $h\neq0$ then the statement of our theorem was proved,
see Theorem~21 in \cite{Combot497649}. Thus, we assume that $c\neq 0$. For given
positive masses $m_1, \ldots, m_n$ we take the Euler--Moulton collinear central
configuration~$\vs$, and  the Kepler homographic solution $\vr(t)$ with chosen
values of $(h,c)$, corresponding to  this configuration, see
Proposition~\ref{pro:homographic}. Clearly, $(\vr(t),\vvv(t))$ defines a
particular solution in $\scM_{h,c}$.  The variational equation in variables
$(\vX, \vX')$, see~\eqref{eq:14}, splits into four-dimensional subsystems
corresponding to variables $(\vY_k,\vY_k')$, see~\eqref{eq:vsplit}.  

If the system is integrable on $\scM_{h,c}$, then, according to
Corollary~\ref{hardcore}, the identity component of the differential Galois
group of variational equation ~\eqref{eq:14} is solvable.  Hence,  the
differential Galois group of each four-dimensional subsystem~\eqref{eq:vsplit}
has a solvable identity component. 

We consider one equation~\eqref{eq:vsplit} with $\delta_k>0$ that corresponds to
$\lambda_k>1$, and we rewrite it as a system of first-order equations
\begin{equation}
\label{eq:vifirst}
\vz' = \vB(\nu) \vz , \qquad \vB(\nu)=\begin{bmatrix}
\vzero & \vI_2 \\
\dfrac{1}{f(\nu)}\vG & -2\vJ_2
\end{bmatrix},
\end{equation} 
where $\vz = (\vY, \vV)$.  Because the
index $k$ is fixed, we simply omit it. 

\subsection{Case $2hc^2\neq-\mu^2$}

First we consider a real solution of the Kepler problem~\eqref{eq:kep2} given by 
\begin{equation}
\rho(\nu)= \frac{c^2/\mu}{1+e\cos \nu},   \qquad 2hc^2 = \mu^2(e^2-1). 
\end{equation}
For this solution $f(\nu)=1+e\cos \nu$, and if $e=0$, then
equation~\eqref{eq:vifirst} is autonomous, hence it does not give any
obstruction for the integrability. Thus, we assume here that $e>0$, or,
equivalently, that $2hc^2\neq -\mu^2$.

We show that  equation~\eqref{eq:vifirst} splits into a direct
product of two-dimensional sub-systems 
\begin{equation}
\label{eq:map2}
\vu' = \vB_\pm(\nu)\vu, \qquad \vu=(u_1, u_2),
\end{equation}
where matrices $\vB_\pm(\nu)$ have the form 
\begin{equation}
\label{eq:Bpm}
\vB_\pm(\nu) = \begin{bmatrix}
-\dfrac{e (2 \delta +3) \sin (\nu )}{f(\nu) g(\nu)}  & 
\dfrac{b_{12} \pm \Delta }{4 f(\nu) g(\nu)}  \\
-\dfrac{b_{21}\pm \Delta }{4 f(\nu) g(\nu)} &
-\dfrac{e (\delta +2f(\nu)) \sin (\nu )}{f(\nu) g(\nu)} 
\end{bmatrix},
\end{equation}
where $g(\nu)= 4+ 3\delta+e\cos \nu$, and 
\begin{align}
b_{12}&= 4 e^2 \cos (2 \nu )+2 e^2+8 e (\delta +2) \cos (\nu )-
9 \delta  (\delta +1)+4, \nonumber \\ 
b_{21}&=2 e^2+8 e (2 \delta +3) \cos (\nu )+(3 \delta +4) (3 \delta +7),\nonumber \\
\Delta^2 &= 4 e^4+4 e^2 (\delta +1) (\delta +4)+(\delta +1) (3 \delta +4)^2 (9 \delta +1).
\label{eq:Delta2}
\end{align}
Notice that $\Delta^2>0$ for positive $\delta$ and $e$. We obtain desired splitting  
making transformation $\vz\mapsto \vT(\nu)\vz$ with  matrix 
\begin{equation}
\vT(\nu):= 
\begin{bmatrix}
\vI_2 & \vI_2 \\
\vB_+(\nu) &\vB_-(\nu)
\end{bmatrix}, \qquad \det  \vT(\nu) = \frac{\Delta^2}{4 f(\nu)^2 g(\nu)^2 }.
\end{equation} 
In effect we obtain system 
\begin{equation}
\vz' = \widetilde{\vB}(\nu) \vz, \qquad   \widetilde{\vB}(\nu)=
\vT(\nu)^{-1}\left(  \vB(\nu)  \vT(\nu)  -\vT'(\nu) \right),
\end{equation}
where 
\begin{equation}
\label{eq:tBpm}
\widetilde{\vB}(\nu)=
\vT(\nu)^{-1}\left(  \vB(\nu)  \vT(\nu)  -\vT'(\nu) \right) = 
\begin{bmatrix}
\vB_+(\nu) & \vzero \\
\vzero & \vB_-(\nu)
\end{bmatrix}.
\end{equation}
Therefore, our claim is proved. The described transformation is just a simple
generalization of a similar transformation found for the variational equations
around libration points of the elliptic restricted three-body problem by
J.~Tschauner in~\cite{Tschauner:1971::,Tschauner:1971::a}.

For our further analysis, we take equation~\eqref{eq:map2} with matrix $\vB_+(\nu)$.  We rewrite it as a second-order scalar differential equation for the function $u= u_2$.  This equation reads 
\begin{equation}
\label{eq:sec_nu} 
u'' + a_1(\nu)u' + a_0(\nu) u=0,
\end{equation} 
with coefficients
\begin{equation}
\begin{split}
a_1(\nu) =& \frac{8 (2 \delta +3) e \sin\nu }{h(\nu)}, \\
a_0(\nu) = & \frac{2 \left(-5 (\delta +1) (3 \delta +4)+\Delta +2 e^2\right)}{h(\nu)}+\frac{\delta}{f(\nu)}+2,
\end{split}
\end{equation}
where
\begin{equation*}
h(\nu)=9 \delta ^2+33 \delta +\Delta +2 e^2+8 (2 \delta +3) e \cos \nu +28.
\end{equation*}

By assumptions $e>0$, therefore we can introduce a new independent variable $z =
1 + e\cos \nu$  in equation~\eqref{eq:sec_nu}. We obtain  
\begin{equation}
\label{eq:sec_zz}
u''  + b_1(z) u' + b_0(z) u=0,
\end{equation} 
where
\begin{equation}
\begin{split}
b_1(z)=&\frac{z-1}{(z-e-1) (z+e-1)}-\frac{8 (2 \delta +3)}{l(z)},
\\
b_0(z) =&- \frac{1}{z(z-e-1) (z+e-1) l(z)}\Big[\delta  \left(\delta  (9 \delta +17)+\Delta +2 e^2+4\right)\\
&+4 z \left((\delta -3) \delta +\Delta +2 e^2-8\right)+16 (2 \delta +3)
z^2\Big],
\end{split}
\label{eq:bbz}
\end{equation}
with
\begin{equation*}
l(z)= 8 (2 \delta +3) z+ \delta  (9 \delta +17)+\Delta +2 e^2+4.
\end{equation*}
The prime in~\eqref{eq:sec_zz} denotes the differentiation with respect to $z$.

The following lemma plays a key role and its proof is given in Appendix~\ref{append:main}. 
\begin{lemma}
	\label{lem:main}
	If $e>0$, and $\delta>0$, then the identity component of the differential Galois group over the field $\CC(z)$ of equation \eqref{eq:sec_zz} with coefficients \eqref{eq:bbz} is not solvable. 
\end{lemma}
From this lemma, we immediately deduce.
\begin{corollary}
	\label{cor:c2}
	If  $2hc^2\neq-\mu^2$, then the planar $n$-body problem, with $n>2$,
	restricted to the level $\scM_{h,c}$, is not integrable.
\end{corollary} 
\begin{proof}
Our proof is based on Corollary~\ref{hardcore}. The particular solution is given
by the Keplerian homographic solution $(\vr(r),\vp(t))$. Therefore, the base
differential field is $K=\CC(\vr(t),\vp(t))$.  We have to show that over this
field the identity component of differential Galois group of variational
equations is not solvable.   To show this, we make  two transformations. The
first transcendental transformation is  $t \mapsto \nu$, where $\nu$ is the true
anomaly.  As a result the differential field is $ \widetilde
K=\CC(\vr(t(\nu)),\vp(t(\nu))=\CC(\cos(\nu),\sin(\nu)) $. In our final step we
made transcendental transformation  $\nu\mapsto z = 1 + e \cos(\nu)$. In effect,
we get differential field $L$ which is an algebraic extension of $\CC(z)$ of
degree $2$, as $\sin(\nu)=\sqrt{1-\cos(\nu)^2}$. Therefore, the composition of
these transformations is an algebraic variable change, and so preserves the
identity component of the differential Galois group of considered variational
equations.   
Hence, Lemma~\ref{lem:main} proves our claim.
\end{proof}

\subsection{Case $2hc^2=-\mu^2$}
The energy integral~\eqref{eq:hc}  expressed in the variable $f(\nu)=c^2/(\mu\rho(\nu))$ reads 
\begin{equation}
\label{eq:hfn}
h = \frac{\mu^2}{2 c^2}\left[f'(\nu)^2 +  (f(\nu)-1)^2 - 1\right].  
\end{equation} 
Thus, if $2hc^2=-\mu^2$, then 
\begin{equation}
\label{eq:df0}
f'(\nu )= \pm\rmi (f(\nu)-1), 
\end{equation}
and we can take $f(\nu)= 1 + \rme^{\rmi\nu}$ as a solution of this equation.
With this choice,  equation~\eqref{eq:vifirst} splits into a direct product of
two-dimensional subsystems~\eqref{eq:map2} with 
\begin{equation}
\label{eq:Bpme0}
\vB_\pm(\nu) = \begin{bmatrix}
\dfrac{\rmi \rme^{\rmi \nu} (2 \delta +3)}{f(\nu) g(\nu)}  & 
\dfrac{b_{12} -(3\delta+4) (3\delta-1\pm\Delta)}{4 f(\nu) g(\nu)}  \\[1em]
\dfrac{b_{21} - (3\delta+4) (3\delta+7\mp\Delta) }{4 f(\nu) g(\nu)} &
\dfrac{\rmi \rme^{\rmi \nu}( \delta +2f(\nu))}{f(\nu) g(\nu)} 
\end{bmatrix},
\end{equation}
where $g(\nu)= 4+ 3\delta+\rme^{\rmi \nu}$, and 
\begin{align}
b_{12}&= 8 (\delta +2) \rme^{\rmi \nu }+8 \rme^{2 \rmi \nu }, \nonumber \\ 
b_{21}&=-8 (2 \delta +3) \rme^{\rmi \nu },\nonumber \\
\Delta^2&= 9\delta^2 +10\delta+1.
\label{eq:Delta22}
\end{align}
Similarly to the previous case,  the gauge map $\vz\mapsto \vT(\nu)\vz$  defined
by   the matrix 
\begin{equation}
\vT(\nu):= 
\begin{bmatrix}
\vI_2 & \vI_2 \\
\vB_+(\nu) &\vB_-(\nu)
\end{bmatrix}, \qquad \det  \vT(\nu) = 
\frac{(3\delta+4)^2\Delta^2}{4 f(\nu)^2 g(\nu)^2 },
\end{equation} 
transforms  the system~\eqref{eq:vifirst}  to the form~\eqref{eq:map2}. 

For our further  analysis,  we take equation~\eqref{eq:map2} with matrix
$\vB_+(\nu)$ defined by~\eqref{eq:Bpme0}.  We introduce a new independent
variable $ z = 1 + \rme^{\rmi\nu}$, and  rewrite it as a scalar second order
differential equation for the function $u= u_2$. It has the form
\begin{equation}
\label{eq:sec0}
u'' + b_1(z) u' + b_2(z)=0, 
\end{equation}
where 
\begin{equation}
\label{eq:b1b2}
\begin{split}
b_1(z) = & \frac{1}{z-1}-\frac{8 (2 \delta +3)}{l(z)}, \\ 
b_2(z) = & \frac{1}{(z-1)^2}\left[ \frac{8 (2 \delta +3) (3 \delta +2 z+2)}{l(z)}-\frac{\delta }{z}-4\right], \\
l(z)= &\delta  (9 \delta -3 \Delta +17)-4 \Delta +8 (2 \delta +3) z+4.
\end{split}
\end{equation}
\begin{lemma}
	\label{lem:main0}
	If $\delta>0$, then the identity component of the  differential Galois group over the field $\CC(z)$ of
	equation \eqref{eq:sec0} with coefficients \eqref{eq:b1b2} is not solvable. 
\end{lemma}
We prove this lemma in Appendix~\ref{append:main0}. Using it, we show the
following. 
\begin{corollary}
	\label{cor:c3}
	If  $2hc^2=-\mu^2$, then the planar $n$-body problem, with $n>2$, restricted
	to the level $\scM_{h,c}$, is not integrable.
\end{corollary}
Finally, Corollary~\ref{cor:c2} and Corollary~\ref{cor:c3} prove our Theorem~\ref{thm:we}.  
\appendix
\section{Second order differential equations}
\label{app:B}

We recall basic notions and facts concerning linear second-order differential
equations relevant to  this paper. This topic is clearly presented in
concise Chapter X  of the book \cite{Whittaker:35::}.

Consider the second-order linear differential equation in reduced form with
rational coefficient
\begin{equation}
w'' =r(z)w,\qquad r(z)\in\mathbb{C}(z).
\label{eq:pom}
\end{equation}
A point $z=c\in\mathbb{C}$ is a singular point of this equation if it is a pole
of $r(z)$. It is a regular singular point if it is a pole of order not greater
than two. Infinity is a regular singular point if function $z^{-2}r(z^{-1})$ has
a pole $z=0$ of order not higher than two.  We assume that all the singularities
of the considered equation are regular. This type of equation is called
Fuchsian.

Then function $r(z)$ has the following expansions: at a singular point $c\in\CC$
\begin{equation}
r(z)=\frac{\alpha_c}{(z-c)^2} + \scO \left((z-c)^{-1}\right)
\label{eq:alphazi}
\end{equation}
and at  infinity
\begin{equation}
r(z)=\frac{\alpha_\infty}{z^2} + \scO \left(z^{-3}\right).
\label{eq:alphazinf}
\end{equation}

Near each singular point $c$ equation~\eqref{eq:pom} has two linearly
independent solutions  of the form
\begin{equation}
w_{\pm}(z)=(z-c)^{\rho_\pm}f(z),\qquad f(z)=1+\sum_{i=1}^{\infty}f_i^{\pm} (z-c)^i,
\label{eq:solFro}
\end{equation}
where $\rho_\pm$ are   solutions of the so-called indicial equation
\begin{equation}
\rho(\rho-1) -\alpha_c = 0,
\end{equation}
that is
\begin{equation}
\rho_\pm=\frac{1}{2}\left(1 \pm\sqrt{1+4\alpha_c}\right).
\end{equation}
They are called exponents of equation \eqref{eq:pom} at singularity $c$.  Then
substituting the expansion~\eqref{eq:solFro} into equation~\eqref{eq:pom} we
successively determine the coefficient $f_i^{\pm}$ for $i=1,2, \dots$.

This statement is valid under the assumption that the difference of exponents
$\Delta_c=\rho_+-\rho_-= \sqrt{1+4\alpha_c}$ is not an integer. If $s=\Delta_c$
is a nonnegative integer, then the solution $w_+(z)$ has the prescribed
form~\eqref{eq:solFro}. But the second solution $w_-(z)$ has to be determined using
the variation of constants method. It has the form
\begin{equation}
w(z)= g_s w_+(z)\ln(z-c)+(z-c)^{\rho_-} h(z),
\label{pp}
\end{equation}
where function $h(z)$ is holomorphic at $z=c$. The constant coefficient $g_s$
can be determined using the method described in \cite{Whittaker:35::}, and an
explicit formula is given in \cite{mp:13::e}. For small $s$ they are
\begin{equation}
\label{eq:gs}
g_0=1, \qquad g_1= 2f_1^+, \qquad g_2= -2f_2^+ + 3(f_1^+)^2.
\end{equation}
If for a given singularity, logarithmic terms are present in a local  solution,
then we say that this singularity is logarithmic. 

The differential Galois group of equation~\eqref{eq:pom} is a subgroup of
$\mathrm{SL}(2,\CC)$, see, for example, \cite{Kovacic:86::,Morales:99::c}. To test the integrability, we need to decide if it is a proper subgroup of
$\mathrm{SL}(2,\CC)$.  The presence of a logarithmic singularity allows one to
formulate a simple and effective criterion.  Before its formulation, we recall
that the function $w(z)$ is called hyperexponential if $w'(z)/w(z)$ is a rational
function. In other words, it is of the form $w(z)=\exp[\int a(z) \rmd z]$, where
$a(z)$ is a rational function.
\begin{lemma}
	\label{lem:log}
	Let us assume that the differential equation~\eqref{eq:pom} has a logarithmic singularity
	and has no non-zero hyperexponential solution. Then its differential Galois group over $\CC(z)$ 
	is $\mathrm{SL}(2,\CC)$.
\end{lemma}
Let us briefly describe the algorithm that allows us to find a hyperexponential
solution of the Fuchsian equation~\eqref{eq:pom} if it exists. In fact, it is
the first case of the Kovacic algorithm \cite{Kovacic:86::}.

The hyperexponential solution we are  looking for has the following form
\begin{equation}
\label{eq:h-exp}
w(z) = P(z)\exp \left[ \int \omega(z) \rmd z\right],\qquad \omega(z)=\sum_{i=1}^M \frac{\varepsilon_i}{z-z_i},
\end{equation}
where $P(z)$ is a polynomial, $z_i$ are poles of $r(z)$ of the first or the second order and $\varepsilon_i$ is an exponent corresponding to this
singularity  (for each point $z_i$ we have two choices for $\varepsilon$). The
polynomial $P(z)$ has degree $d$ where
\begin{equation}
d = \varepsilon_\infty -\sum_{i=1}^M \varepsilon_i,
\end{equation}
and $\varepsilon_\infty$ is an exponent at infinity. Here, the sum is taken over
all poles of $r(z)$, and in the case of a first-order pole, it is assumed that
there is only one exponent equal to one.  Moreover, the polynomial $P(z)$ must
be a solution of the following equation
\begin{equation}
\label{eq:pol_kov}
P'' +2 \omega(z) P' + (\omega'(z) + \omega(z)^2 - r(z)) P = 0.
\end{equation}
In practice, as a candidate for $P(z)$ one can take a polynomial of degree $d$
with indefinite coefficients and insert it into the above equation. The problem
then reduces to finding solutions of a system of linear equations. For an equation with $m$
singular points,  we have $2^{m+1}$ possible choices of elements
$(\varepsilon_1,\ldots,\varepsilon_m, \varepsilon_\infty)$, so we have at most
such a number of possibilities for the choice of $d$ and $\omega(z)$.

Proof of Lemma~\ref{lem:log} and its numerous applications can be found in
\cite{Morales:99::c} and \cite{Boucher:00::,Boucher:06::}.

\section{Proof of Lemma~\ref{lem:main}}
\label{append:main}

We consider equation 
\begin{equation}
\label{eq:sec_z}
u''  + b_1(z) u' + b_0(z) u=0,
\end{equation} 
where  $ b_0(z)$ and $b_1(z)$ are rational functions given by~\eqref{eq:bbz}.
We transform this equation to the reduced form
\begin{equation}
\label{eq:red+}
v'' = r(z) v, \qquad r(z) =\frac{1}{2}b_1'(z)  + \frac{1}{4}b_1(z)^2 -b_0(z),
\end{equation}
by setting
\begin{equation}
\label{eq:tran}
u = v \exp\left[ -\frac{1}{2} \int^z_{z_\ast} b_1(s)\, \rmd s \right].
\end{equation}
Rational function $r(z)$ reads
\begin{equation}
\label{eq:rr+g}
\begin{split}
& r(z)=-\frac{\delta }{\left(e^2-1\right) z}-\frac{3}{16 (z-e-1)^2}-\frac{3}{16 (z+e-1)^2}+\frac{3}{4\left(z-z_3 \right)^2}\\
&-\frac{6 \delta +5 z+3}{2 (z-e-1)
	(z+e-1) \left(z-z_3\right)}+\frac{33
	\left(e^2-1\right)+8 \delta 
	(z-2)}{8 \left(e^2-1\right)
	(z-e-1) (z+e-1)}, 
\end{split}
\end{equation}
where 
\begin{equation}
z_3 =-\frac{\delta  (9\delta +17)+\Delta +2 e^2+4}{8(2 \delta+3)}.
\end{equation}
We write it also in the form $r(z)=P(z)/\widetilde{Q}(z)$ where the denominator
$\widetilde{Q}(z)$ is 
\begin{equation}
\label{eq:Q(z)}
\begin{split}
\widetilde{Q}(z)&=4 (e-z+1) (e+z-1)\times \\&\left[\delta  (9 \delta +17)+\Delta +2 e^2+8 (2 \delta +3)z+4\right]Q(z),\\
Q(z)&=z (e-z+1) (e+z-1) \left[\delta  (9 \delta +17)+\Delta +2 e^2+8 (2 \delta +3) z+4\right].
\end{split}
\end{equation}
In other words, $ Q(z)$ contains linear in $z$ factors of the denominator.
Singular points of equation~\eqref{eq:red+} with  $r(z)$ given in
\eqref{eq:rr+g} are roots of  $Q(z)$, and the infinity, thus
\begin{equation}
\label{eq:si}
z_0= 0, \quad
z_{1,2}=1 \mp e, \quad z_{3}, \quad z_{\infty}=\infty.
\end{equation}
At first, we consider  the generic case, that is we assume 
\begin{itemize}
	\item all five singular points are pairwise different, and 
	\item $z_0$ is  a pole of the first order, and all remaining singular points
	are poles of the second order. 
\end{itemize}
From  \eqref{eq:rr+g}  we get  expansions  
\begin{equation}
\label{eq:sf}
r(z) = \frac{\alpha_i}{ (z -z_i)^2 } + \scO((z-z_i)^{-1}),
\end{equation}
where
\begin{equation*}
\alpha_0=0,\quad \alpha_{1}=\alpha_{2}=-\frac{3}{16},\qquad \alpha_3=\frac{3}{4}.
\end{equation*}
Moreover,  expansion of  $r(z)$ at infinity is
\begin{equation}
r(z) = \frac{\alpha_{\infty}}{z^2} +\scO(z^{-3}),\qquad \alpha_{\infty}=2.
\label{eq:infinitek}
\end{equation}
The differences of exponents  $\Delta_i$ at the respective singularities according to formula $\Delta_i=\sqrt{1+4\alpha_i}$ are
\begin{equation*}
\Delta_0=1,\quad \Delta_1=\Delta_2=\frac{1}{2}\quad \Delta_3=2,\quad \Delta_{\infty}=3.
\end{equation*}
Since the differences of exponents of local solutions around singularities
$z_0$, $z_3$, and $z_{\infty}$ are integers, one can expect logarithmic terms.
Using formulae~\eqref{eq:gs} we get that if $\delta\neq0$ then there is a local
solution around $z_0$  containing the logarithmic term.

Now, if we show that equation~\eqref{eq:red+} has no non-zero
hyperexponential solution, then, by Lemma~\ref{lem:log}, we conclude the proof
in the considered case. In this aim we have to select all possible choices of
exponents $\vveps=(\varepsilon_0,
,\varepsilon_1,\varepsilon_2,\varepsilon_3,\varepsilon_\infty)$ such that
\begin{equation*}
d = d(\vveps)= \varepsilon_\infty - (\varepsilon_0+ \varepsilon_1+\varepsilon_2+\varepsilon_3) \in \N_0,
\end{equation*}
where $\N_0$ is the set of  non-negative integers. Sets of exponents of local solutions at singularities are
\begin{equation*}
E_0=\{1\},\quad E_1=E_2=\left\{ \frac{1}{4}, \frac{3}{4}\right\},\quad E_3=\left\{ -\frac{1}{2}, \frac{3}{2}\right\},\quad
E_{\infty}=\left\{ -1,2\right\}.
\end{equation*}
We have only two admissible choices of $\vveps$ 
\begin{equation*}
\begin{split}
& \vveps_{1}=  \left(1, \frac{3}{4}, \frac{3}{4},  -\frac{1}{2}, 2 \right)  ,\qquad d\left(\vveps_{1}\right)=0,\\
& \vveps_{2}=  \left(1, \frac{1}{4}, \frac{1}{4},  -\frac{1}{2}, 2 \right)  ,\qquad d\left(\vveps_{2}\right)=1.
\end{split}
\end{equation*}
We have to check whether for a given $\vveps_i$ the respective equation~\eqref{eq:pol_kov} has a polynomial solution $P(z)$  of degree $d(\vveps_{i})$. 
For $ \vveps_1$ we can assume that $P(z)=1$. It is a solution of~\eqref{eq:pol_kov} iff
\begin{equation*}
(\delta +1) (2 \delta +3)=0,\quad -\delta  (\delta  (9 \delta +44)+\Delta +71)-3 (\Delta +12)+2 (7 \delta +9) e^2=0.
\end{equation*}
As $\delta>0$, the first equation does not have a solution.
For  $ \vveps_2$ we take  $P(z)=z+s_1$, and  it is a solution of~\eqref{eq:pol_kov} iff
\begin{equation*}
\begin{split}
&(2 \delta +3) (\delta -s_1)=0,\\
&s_1\left(5 \delta  (\delta +9)-3 \Delta -6 e^2+60\right)-(\delta +6) \left(\delta  (9
\delta +17)+\Delta +2 e^2+4\right)=0,\\
& s_1 \left(-\delta  (\delta  (9 \delta +26)+\Delta +37)-\Delta +2 (7 \delta +11)
e^2-28\right)-2 (e^2-1) \Big(\delta  (9 \delta +17)\\
&   +\Delta +2 e^2+4\Big) =0. 
\end{split}
\end{equation*}
The first equation gives $s_1=\delta$, and  then the remaining two equations simplify to
\begin{equation*}
\begin{split}
&-2 (2 \delta +3) \left(\delta  (\delta +5)+\Delta +2 e^2+4\right)=0,\\
&-\delta  (\delta  (\delta  (9 \delta +26)+\Delta
+19)+\Delta -6)+2 (\Delta +4)-4 e^4-2 e^2 (2 \delta  (\delta +3)+\Delta +2)=0.
\end{split}
\end{equation*}
Since $\delta,\Delta$ and $e$ are non-negative, this system does not have a solution.
Thus, we conclude that equation \eqref{eq:red+} has no hyperexponential solution
and the differential Galois group is $\mathrm{SL}(2,\CC)$.

Now, we consider non-generic cases. If a confluence of singularities occurs,
then the discriminant of  polynomial $Q(z)$ given in~\eqref{eq:Q(z)} 
\begin{equation}
\label{eq:gene1}
\begin{split}
\mathrm{disc}(Q(z))=&
4 e^2 (e-1)^2 (e+1)^2  \left[\delta  (9 \delta +17)+\Delta +2 e^2+4\right]^2\times\\ &\left[9 \delta ^2+33 \delta +\Delta +16 \delta  e+2 e (e+12)+28\right]^2 \times\\
&\left[9 \delta ^2+33 \delta +\Delta
+2 e^2-8 (2 \delta +3) e+28\right]^2
\end{split}
\end{equation}
vanishes. As $e$, $\delta$ and $\Delta$ are positive, we have two possibilities: $e=1$ or $9 \delta ^2+33 \delta +\Delta+2 e^2-8 (2 \delta +3) e+28=0$.
Taking into account the definition of $\Delta$, see~\eqref{eq:Delta2}, we obtain a
condition that, after division by term $(2 \delta +3)>0$, reads 
\begin{equation}
(e-1)(e-3\delta-4) (2 e-3 \delta-4)=0.
\label{eq:inki}
\end{equation}
A non-generic case can occur also when the numerator and the denominator of
$r(z)$ have a common factor that cancels and the orders of poles can change. To
verify this, we calculated the resultant of $P(z)$ and $Q(z)$. It has the form 
\begin{equation}
\operatorname{result}\left(P(z),Q(z)\right)=
\delta  e^4 (e-1)^2 (e+1)^2 (2 \delta +3)^2 S_1S_2S_3S_4,
\end{equation}
where $S_i$ are polynomials in $(e,\delta, \Delta )$. By assumption, all these
variables are positive. Hence, either $e=1$, or $S_i=0$, for a certain $i$.
However, the polynomials $S_i$ are linear in $\Delta$, so, using the
definition~\eqref{eq:Delta2} we can transform condition $S_i=0$ to the form
$S(e,\delta)=0$. Polynomial $S(e,\delta)$ factors, and it vanishes for positive
$e$ and $\delta$  only if condition~\eqref{eq:inki} is fulfilled. In effect, we
have to consider only three non-generic cases when a particular factor in \eqref{eq:inki} vanishes.
Assume that $e=1$. Then $\Delta=\delta  (9 \delta +17)+6$, and $r(z)$ is
\begin{equation}
\begin{split}
r(z)&=-\frac{8 \delta +3}{16 z^2}-\frac{3}{16 (z-2)^2}+\frac{3}{4 \left(\frac{\delta  (9 \delta
		+17)+6}{8 \delta +12}+z\right)^2}\\
&+\frac{\delta  (\delta  (36 \delta +173)+201)+4 (2 \delta +3) (4 \delta +13) z+54}{8 (z-2) z (\delta  (9 \delta +17)+4 (2 \delta +3) z+6)}.
\end{split}
\end{equation}
The expansion of $r(z)$ at infinity still has the form \eqref{eq:infinitek}. The
variational equation \eqref{eq:red+} has four singularities: $z_0=0$, $z_1=2$,
$z_2=-\frac{\delta  (9 \delta +17)+6}{8 \delta +12}$, and infinity. All
singularities are regular. One can check that the discriminant of the
denominator of $r(z)$ is positive, so there is no confluence of singularities.
Moreover, the resultant of the numerator and the denominator of $r(z)$ is also
positive, so they have no common factor. 

The differences of exponents  $\Delta_i$ at the respective singularities are
\begin{equation}
\Delta_0=\frac{1}{2} \sqrt{1-8 \delta },\quad \Delta_1=\frac{1}{2},\quad \Delta_2=2,\quad \Delta_{\infty}=3.
\end{equation} 
There is a local solution at infinity with a logarithmic term. Using
formulae~\eqref{pp}, we find that the multiplier of the logarithmic term is
$g_3=\frac{1}{9} \delta  (\delta +1) (2 \delta +3)>0$. Hence, we will proceed
as in the generic case, and we will prove that the variational equation has no
hyperexponential solution. The sets of exponents are
\begin{equation*}
\begin{split}
E_0&=\left\{\frac{1}{4} \left(\sqrt{1-8 \delta }+2\right),\frac{1}{4} \left(2-\sqrt{1-8 \delta }\right)\right\},\quad E_1=\left\{ \frac{1}{4}, \frac{3}{4}\right\},\quad E_2=\left\{-\frac{1}{2},\frac{3}{2}\right\},\\
E_{\infty}&=\left\{ -1,2\right\}.
\end{split}
\end{equation*}
If $d =\varepsilon_\infty - (\varepsilon_0+ \varepsilon_1+\varepsilon_2) \in
\N_0$,  where $\varepsilon_i\in E_i$, then 
\begin{equation*}
\sqrt{1-8 \delta }=4d+s,\quad s\in\{-7,-5,1,3,5,7,13,15\}.
\end{equation*}
However, $\delta>0$ and $4d+s$ are real, so
$\delta\in\left[\tfrac{1}{8},0\right)$. Hence, $1>\sqrt{1-8 \delta }\geq0$, but
if $4d+s>0$, then non-negative $4d+s\geq 1$. Thus, we conclude that equation
\eqref{eq:red+} in the considered case has no hyperexponential solution, and the
differential Galois group is $\mathrm{SL}(2,\CC)$.

In the second non-generic case, we assume that $e-3\delta-4=0$,
see~\eqref{eq:inki}. Then  we have   $\Delta = (3 \delta +4) (7 \delta +9)$, and
$r(z)$ simplifies to
\begin{equation}
\begin{split}
r(z)&=-\frac{\delta }{3 (\delta +1) (3\delta +5) z}-\frac{3}{16 (z+3+3\delta)^2}-\frac{3}{16 (z-5-3
	\delta)^2}\\
&+\frac{171 \delta ^2+8 \delta (z+55)+285}{24 (\delta +1) (3
	\delta +5) (z+3+3\delta) (z-5-3\delta)}.
\end{split}
\end{equation}   
The expansion of $r(z)$ at infinity still has the form \eqref{eq:infinitek}. Now
equation \eqref{eq:red+} has four singularities:  $z_0=0$, $z_1=-3 (\delta +1)$
and $z_2=3 \delta +5$ and infinity. The singularity $z_0$ is a pole of the first
order, $z_1$ and $z_2$ are poles of the second order, and infinity is
regular.  The discriminant of the denominator of $r(z)$ is positive; thus there
is no coalescence of singularities. In addition, the resultant of the numerator and the
denominator of $r(z)$ does not vanish, so they do not have a common factor. Differences of exponents $\Delta_i$ at the respective singularities are
\begin{equation}
\Delta_0=1,\quad \Delta_1=\Delta_2=\frac{1}{2},\quad \Delta_{\infty}=3,
\end{equation} 
so,  the sets of exponents are
\begin{equation*}
E_0=\{1\},\quad E_1=E_2=\left\{ \frac{1}{4}, \frac{3}{4}\right\},\quad
E_{\infty}=\left\{ -1,2\right\}.
\end{equation*}
One can check that there is a local solution of \eqref{eq:red+} near
$z=0$ with a logarithmic term.

We have only two choices of $\vveps=(\varepsilon_0,
,\varepsilon_1,\varepsilon_2,\varepsilon_\infty)$ such that $d(\vveps)=
\varepsilon_\infty - (\varepsilon_0+ \varepsilon_1+\varepsilon_2)$ $ \in \N_0$,
namely 
\begin{equation*}
\begin{split}
& \vveps_1=  \left(1, \frac{3}{4}, \frac{1}{4},   2 \right)  ,\quad d\left(\vveps_1\right)=0,\\
& \vveps_2=  \left(1, \frac{1}{4}, \frac{3}{4},   2 \right)  ,\quad d\left(\vveps_2\right)=0.
\end{split}
\end{equation*}
It is easy to check that for both choices the respective
equation~\eqref{eq:pol_kov} has no non-zero polynomial solution of degree zero.
In effect,  equation \eqref{eq:red+} in the considered case has no
hyperexponential solution, and the differential Galois group over $\CC(z)$ is
$\mathrm{SL}(2,\CC)$.

In the last non-generic case, we assume that $-3 \delta +2 e-4=0$. Thus,  $e=
\frac{3 \delta }{2}+2$ and $\Delta = \frac{1}{2} (3 \delta +4) (7 \delta +6)$.
The rational function $r(z)$ takes the form
\begin{equation*}
\begin{split}
r(z)&=-\frac{4 \delta }{3 \left(3 \delta ^2+8 \delta +4\right) z}-\frac{3}{16 \left(z-\frac{3 (\delta
		+2)}{2}\right)^2}+\frac{5}{16 \left(z+\frac{1}{2} (3 \delta +2)\right)^2}\\
&+\frac{\delta  (135 \delta +32 z+296)+180}{24 (\delta +2) (3 \delta +2)
	\left(z-\frac{3 (\delta +2)}{2}\right) \left(z+\frac{1}{2} (3 \delta +2)\right)}. 
\end{split}
\end{equation*}
Hence, equation \eqref{eq:red+}  has four singularities: $z_0=0$,
$z_1=-\frac{1}{2} (3 \delta +2)$, $z_2=\frac{3 (\delta +2)}{2}$ and
$z_{\infty}=\infty$. As in the previous cases for  $\delta>0$ all their
singularities are pairwise different.  The expansion of $r(z)$ at infinity still has the form \eqref{eq:infinitek}. The differences of exponents  $\Delta_i$ at
the respective singularities are
\begin{equation}
\Delta_0=1,\quad \Delta_1=\frac{3}{2},\quad \Delta_2=\frac{1}{2},\quad \Delta_{\infty}=3,
\end{equation} 
so,  the sets of exponents are
\begin{equation*}
E_0=\{1\},\quad E_1=\left\{ -\frac{1}{4}, \frac{5}{4}\right\},\quad E_2=\left\{ \frac{1}{4}, \frac{3}{4}\right\},\quad
E_{\infty}=\left\{ -1,2\right\}.
\end{equation*}
We have only one choice of exponents
\begin{equation*}
\vveps=  \left(1, -\frac{1}{4}, \frac{1}{4},   2 \right)  ,\quad d\left(\vveps\right)=1.
\end{equation*}
The respective equation~\eqref{eq:pol_kov} has a polynomial solution $P(z)=z+s_1$ iff 
\begin{equation*}
\delta -s_1=0,\qquad  (\delta +4) s_1-3 (\delta +2) (3 \delta +2)=0.
\end{equation*}
The first equation gives $s_1=\delta$ and substitution to the second equation
gives the equality $2 \delta ^2+5 \delta +3=0$ which is not valid for
$\delta>0$. Thus, equation \eqref{eq:red+} in the considered case has no
hyperexponential solution and the differential Galois group over $\CC(z)$ is
$\mathrm{SL}(2,\CC)$. 

In this way, we proved the lemma.

\section{Proof of Lemma~\ref{lem:main0}}
\label{append:main0}
The reduced form of equation~\eqref{eq:sec0} is 
\begin{equation}
w''=r(z) w, \qquad 
r(z)= \frac{48 (2 \delta +3)^2}{l(z)^2}-  \frac{4
	(2 \delta +3) (6 \delta +5 z+3)}{(z-1)^2 l(z)}+\frac{4 \delta +15 z}{4 (z-1)^2 z},
\label{eq:rre0}
\end{equation}
where 
\begin{equation*}
l(z)= \delta  (9 \delta -3 \Delta +17)-4 \Delta +8 (2 \delta +3) z+4.
\end{equation*}
This equation has singularities $z_0=0$, $z_1=1$, root of equation $l(z)=0$,  that
is $z_2=\tfrac{3 \delta \Delta+4 \Delta-9 \delta ^2 -17 \delta -4}{8 (2 \delta
+3)}$, and $z_{\infty}=\infty$. All these singularities are generally different except for particular values of $\delta$, which will be considered later. In the
generic case $z_0$ is a pole of the first order, the remaining poles are  of
the second order. The coefficients $\alpha_i$ of the expansions $r(z)$  at
singularities defined in \eqref{eq:alphazi} are
\begin{equation*}
\alpha_0=0,\quad \alpha_1= \frac{8 (2 \delta +3)}{-3 \delta
	+\Delta -7}+\delta
+\frac{15}{4},\quad \alpha_2= \frac{3}{4},\quad\alpha_{\infty}=2.
\end{equation*}
Differences of exponents at singularities are
\begin{equation*}
\Delta_0=1,\quad \Delta_1=2 \sqrt{\frac{8 (2 \delta+3)}{-3 \delta +\Delta-7}+\delta +4},\quad \Delta_2=2,\quad \Delta_{\infty}=3.
\end{equation*}
One can check that in local solutions of variational equations around $z_0$
exists a logarithmic term provided $\delta\neq0$. 

Now we will check if the reduced equation \eqref{eq:rre0}  has a hyperexponential solution. Let us notice using the explicit form of $\Delta$
given in \eqref{eq:Delta22}, allows to rewrite   $\Delta_1$  as
\begin{equation}
\Delta_1=\sqrt{2} \sqrt{1-\delta
	-\sqrt{(1+\delta) (1+9 \delta)}}.
\end{equation}
For $\delta>0$ the expression under the square root is always negative,
therefore exponent $\varepsilon_1=\tfrac{1}{2}(1\pm \Delta_1)$ at $z_1$ is a
complex number.  Thus, we cannot choose exponents
$\vveps=(\varepsilon_0,\varepsilon_1,\varepsilon_2,\varepsilon_{\infty})$ such
that $\varepsilon_{\infty}-(\varepsilon_0+\varepsilon_1+\varepsilon_2)\in\N_0$.
We conclude that the equation does not have a solution, and, therefore,  its differential Galois group over $\CC(z)$ is $\operatorname{SL}(2,\CC)$.

Let us now consider non-generic cases. For this purpose, we write the rational
coefficient $r(z)$ given in \eqref{eq:rre0} in the form
$r(z)=\tfrac{P(z)}{\widetilde{Q}(z)}$, where the denominator takes the form
$\widetilde{Q}(z)=4(z-1)l(z)Q(z)$ and $Q(z)=z(z-1)l(z)$ contains factors of the
denominator of $r(z)$ linear in $z$. Looking for the coalescence of
singularities, we consider the discriminant of $Q(z)$
\begin{equation*}
\operatorname{disc}(Q(z))=(3 \delta +4)^2 (\delta  (9
\delta -3 \Delta +17)-4
\Delta +4)^2 (-3 \delta
+\Delta -7)^2.
\end{equation*}
As $\Delta^2= 9\delta^2 +10\delta+1$, see the equation \eqref{eq:Delta22}, then the equality  $\delta  (9 \delta -3
\Delta +17)-4  \Delta +4=0$ implies  $-\tfrac{16 \delta  (2 \delta +3)}{(3 \delta
	+4)^2}=0$ , but it is impossible because $\delta>0$. Finally, the equality $(-3 \delta +\Delta -7)=0$ implies $16 (2 \delta
+3)=0$, which is again impossible. To analyze the cases where
$P(z)$ and $Q(z)$ have a common factor that can cancel, we consider the
resultant
\begin{equation*}
\operatorname{result} (P (z), Q (z)) = 
12222 \delta (2 \delta +3) 2 (3 \delta +4) 4S_1 (z) S_2 (z) S_3 (z),
\end{equation*}
where 
\begin{equation*} 
\begin{split}
&S_1(z)= \delta  [\delta  (9 \delta  (9
\delta -3 \Delta +43)-114
\Delta +647)-155 \Delta
+421]-64 (\Delta -1),\\
&S_2(z)=\delta  [\delta  (9 \delta  (9
\delta -3 \Delta +34)-87
\Delta +377)-80 (\Delta
-2)]-16 (\Delta -1),\\
&S_3(z)= \delta  [\delta  (36 \delta -12
\Delta +143)-41 \Delta
+122]-57 \Delta +39.
\end{split}
\end{equation*}
Vanishing of $S_1(z)$ or $S_2(z)$ or $S_3(z)$ lead to conditions
\begin{equation*}
\begin{split}
& -\frac{256 \delta  (2 \delta
	+3)^3}{\left(27 \delta ^3+114
	\delta ^2+155 \delta
	+64\right)^2}=0,\quad \frac{64 \delta ^2 (2 \delta
	+3)^2}{(3 \delta +4)^2
	\left(9 \delta ^2+17 \delta
	+4\right)^2}=0,\\
&-\frac{64 (2 \delta +3)^2
	\left(32 \delta ^2+44 \delta
	+3\right)}{\left(12 \delta
	^2+41 \delta +57\right)^2}=0,
\end{split} 
\end{equation*}
respectively, but these expressions cannot vanish. Thus, there are no non-generic cases to consider.

\section{Proof of Lemma~\ref{lem:quad}}
\label{app:D}

Assume that the quadratic form vanishes in a vector space $\mathcal{W}$. Then it also
vanishes in its complex conjugate $\bar{\mathcal{W}}$. Consider now the
intersection vector space $\mathcal{W}'=\mathcal{W} \cap \bar{\mathcal{W}}$.
Considering  $\mathbb{C}^m$ as $\mathbb{R}^{2m}$, the complex conjugation operator
writes
\begin{equation*}
\sigma = \left[\begin{array}{cc}  \vI_m & 0\\ 0 & -\vI_m \end{array} \right].
\end{equation*}
The vector space $\mathcal{W}'$ is stable by the complex conjugation by
definition and thus is a stable vector space of this matrix. Any stable vector
space of $\sigma$ is a direct sum of vector sub-spaces of each block,
$\mathcal{E}_{+} \oplus \mathcal{E}_{-}$. Thus $\mathcal{W}'$ follows this
decomposition, and we can write a basis $\mathcal{B}$ of it whose elements are
either invariant by $\sigma$, or invariant by $-\sigma$. Considering again 
$\mathcal{W}'$ as a complex vector space, this means that $\mathcal{B}$ has elements that are either real or pure imaginary. Multiplying these latter  by $\rmi$ we get  a real
basis of $\mathcal{W}'$. But then our positive definite quadratic form vanishes
in this real vector space, and thus $\mathcal{W}'=0$. Finally,
\[
\dim(\mathcal{W} + \bar{\mathcal{W}})=\dim\mathcal{W}+\dim\bar{\mathcal{W}} \leq m 
\]
and thus $2\dim \mathcal{W} \leq m$, giving the Lemma.


\section*{Acknowledgments}
The authors wish to express their thanks to Alain Albouy and Christian
Marchal for sending certain hard-to-reach works.  

We are deeply grateful to the reviewers for their valuable comments, which have contributed to the improvement of our article.

For AJM and MP this research has been partially supported by
The National Science Center of Poland under Grant No. 2020/39/D/ST1/01632 and by the Minister of Science under the “Regional Excellence Initiative” program, Project No. RID/SP/0050/2024/1, and the work of  TC  was  founded by  a CNRS IEA grant ``Integrability of Hamiltonian Systems''.

\def\cprime{$'$} \def\cydot{\leavevmode\raise.4ex\hbox{.}} \def\cprime{$'$} \newcommand{\noopsort}[1]{}

\end{document}